\documentclass[12pt]{article}
\usepackage{amsmath}
\usepackage{graphicx}
\usepackage{enumerate}
\usepackage{natbib}
\usepackage{bm}
\usepackage{url} 

\usepackage{amsthm,amsfonts,amssymb}
\usepackage{mathrsfs,braket,enumitem,subcaption,float,xcolor,bbm}
\usepackage[titlenumbered,ruled]{algorithm2e}
\usepackage{setspace}
\usepackage{caption}
\RequirePackage[colorlinks,linkcolor=blue,citecolor=blue,urlcolor=blue]{hyperref}

\newcommand{\R}{\mathbb{R}}
\newcommand{\N}{\mathbb{N}}
\newcommand{\Z}{\mathbb{Z}}
\newcommand{\bc}{\mathbf{c}}

\newcommand{\U}{\mathbf{U}}
\newcommand{\bLambda}{\mathbf{\Lambda}}

\newcommand{\rinitf}[1]{r_{{#1},M}}

\newcommand{\gammaij}{\gamma_{g_i,g_j}}
\newcommand{\gammaji}{\gamma_{g_j,g_i}}
\newcommand{\mE}[2]{F_{#1,#2}}
\newcommand{\m}{\hat{\mu}}

\newcommand{\Mld}[1]{\mathscr{M}_{\infty}({#1})\cap\ell_2(\mathbb{Z})}

\newcommand{\dubar}{ \underline{\delta}}

\newcommand{\vertiii}[1]{{\left\vert\kern-0.25ex\left\vert\kern-0.25ex\left\vert #1\right\vert\kern-0.25ex\right\vert\kern-0.25ex\right\vert}}

\newtheorem{lem}{Lemma}
\newtheorem{thm}{Theorem}
\newtheorem{prop}{Proposition}
\newtheorem{cor}{Corollary}

\usepackage[sectionbib]{bibunits}
\defaultbibliographystyle{plainnat}
\defaultbibliography{bib}

\newif\ifshow 
\showtrue 

\newif\iftwofiles 
\twofilesfalse 

\newcommand{\blind}{0}

\begin{document}

\def\spacingset#1{\renewcommand{\baselinestretch}%
{#1}\small\normalsize} \spacingset{1}

\iftwofiles\begin{bibunit}\fi
\if0\blind
{
  \title{\bf Multivariate moment least-squares variance estimators for reversible Markov chains}
  \author{Hyebin Song\thanks{Both authors contributed equally. The authors gratefully acknowledge NSF support DMS-2311141.} \, and Stephen Berg\footnotemark[1]  \thanks{Corresponding author: sqb6128@psu.edu} \\
Department of Statistics, Pennsylvania State University
}
  \maketitle
} \fi

\if1\blind
{
  \bigskip
  \bigskip
  \bigskip
  \begin{center}
    {\LARGE\bf Title}
\end{center}
  \medskip
} \fi

\bigskip
\begin{abstract}
Markov chain Monte Carlo (MCMC) is a commonly used method for approximating expectations with respect to probability distributions. Uncertainty assessment for MCMC estimators is essential in practical applications. Moreover, for multivariate functions of a Markov chain, it is important to estimate not only the auto-correlation for each component but also to estimate cross-correlations, in order to better assess sample quality, improve estimates of effective sample size, and use more effective stopping rules. \citet{berg2022efficient} introduced the moment least squares (momentLS) estimator, a shape-constrained estimator for the autocovariance sequence from a reversible Markov chain, for univariate functions of the Markov chain. Based on this sequence estimator, they proposed an estimator of the asymptotic variance of the sample mean from MCMC samples. In this study, we propose novel autocovariance sequence and asymptotic variance estimators for Markov
chain functions with multiple components, based on the univariate momentLS estimators from \citet{berg2022efficient}. We demonstrate strong consistency of the proposed auto(cross)-covariance
sequence and asymptotic variance matrix estimators. We conduct empirical comparisons of our method with other state-of-the-art approaches on simulated and real-data examples, using popular samplers including the random-walk Metropolis sampler and
the No-U-Turn sampler from STAN. Supplemental materials for this article are available online.
\end{abstract}

\noindent%
\vfill

\newpage
\spacingset{1.5} 
\section{Introduction}
In modern statistical applications, it is often necessary to analyze intractable probability distributions, for which computing key summary statistics such as expectations or quantiles analytically is challenging. Markov chain Monte Carlo (MCMC) methods offer a practical approach for drawing samples from these intractable distributions. An MCMC experiment involves generating a sequence of random variables $X = (X_0, X_1,\dots)$, such that the distribution of $X_t$ converges to the intractable target distribution $\pi$. This simulated chain $X$ can then be used to estimate the quantity of interest, such as the expectation with regard to the target distribution.  Suppose we have a function $g: \mathsf{X} \to \R^d$ where $d\ge 1$, for which we want to approximate $\mu_g = \int g(x) \pi(dx)$, where $\pi$ is intractable. One can generate a Markov chain $X$ of length $M$ with invariant distribution $\pi$, and use the MCMC estimator $\hat{\mu}_{gM} = M^{-1} \sum_{t=0}^{M-1} g(X_t)$ to estimate $\mu_g$. 

The accuracy of an MCMC estimator can be quantified by its asymptotic variance. Under mild conditions, a Markov chain central limit theorem holds~\citep[e.g.,][]{meyn2009markov}, yielding
\begin{align}\label{eq:Mult_CLT}
    \sqrt{M}(\hat{\mu}_{gM} - \mu_g) \overset{d}{\to} N(0, \Sigma)
\end{align} as $M\to\infty$, where $\Sigma\in\mathbb{R}^{d\times d}$ is the asymptotic variance of the empirical mean $\hat{\mu}_{gM}$.
Due to correlations in $(X_t)_{t\ge 0}$, $\Sigma$ usually differs from $\textrm{Var}_{\pi} (g(X))$. 
In fact, $\Sigma$ aggregates the covariances across all lags, i.e.,
\begin{align}\label{eq:avar_sum}
    \Sigma = \sum_{k=-\infty}^\infty \textrm{Cov}_\pi (g(X_0), g(X_{k})).
\end{align}
where we define $\textrm{Cov}_\pi (g(X_0), g(X_{k})) = \textrm{Cov}_\pi (g(X_{|k|}), g(X_{0}))$ for $k<0$.

Historically, the assessment of uncertainty in MCMC estimates has often been conducted on a component-wise basis. Let $\hat{\mu}_{gM,j}$ and $\mu_{gj}$ denote the $j$th components of $\hat{\mu}_{gM}$ and $\mu_{g}$, for $j=1,...,d$, and let $\Sigma_{ij}$ denote the $i,j$ element of $\Sigma$. Componentwise uncertainty assessments amounts to considering univariate CLT results $\sqrt{M}(\hat{\mu}_{gM,j} - \mu_{gj}) \overset{d}{\to}N(0, \Sigma_{jj})$ for $j=1,\dots,d$, and therefore focusing only on the diagonals of the asymptotic variance $\Sigma$. In particular, the quantity $\sqrt{\Sigma_{jj}/M}$  is called the MCMC standard error of $\hat{\mu}_{gM,j}$. Estimates of these MCMC standard errors have been used to determine effective sample sizes, construct confidence intervals, and implement stopping criteria to achieve a predefined level of precision for each component in the Markov chain \citep{heidelberger1981spectral,jones2006fixed}. 

However, by accounting for dependencies among the components of functions of a Markov chain, we can gain a more comprehensive understanding of the overall quality of the Markov chain estimates and formulate a more effective stopping criterion \citep{dai2017multivariate,vats2019multivariate}. For instance, using an estimate of the asymptotic variance matrix $\Sigma$, we can construct confidence regions for $\mu_g$ with smaller volumes compared to Bonferroni-corrected confidence intervals. Additionally, an effective sample size can be computed to summarize the overall inflation or deflation of sample size relative to an independent and identically distributed (iid) chain. This approach allows for the formulation of stopping rules based on an overall accuracy criterion and offers practical advantages over a component-wise criterion that terminates based on the worst-performing component, which might be overly conservative.

Moreover, the estimation of the entire asymptotic variance matrix $\Sigma$ for a multivariate chain, as opposed to only estimating its diagonals, is essential in quantifying uncertainties in linear combinations of $\hat{\mu}_{gM,j}, j=1,\dots,d$. This is especially useful in devising MCMC estimators with reduced variance using control variate methods~\citep[see, e.g.][]{dellaportas2012control,berg2019control,belomestny2020variance,south2022postprocessing}. 
Control variate methods employ estimators of the form  $\tilde{\mu}_{gM,C}=M^{-1}\sum_{t=0}^{M-1}\{g(X_t)-C^\top h(X_t)\}$, where $h:\mathsf{X}\to\mathbb{R}^{p}$ is a control variate function chosen to satisfy $\int h(x)\pi(dx)=0$, and $C\in\mathbb{R}^{p\times d}$. Since $h$ has $\pi$ expectation 0, the terms $C^\top h(X_t)$ do not lead to bias under stationarity. If $C$ is chosen such that the asymptotic variance of $\tilde{\mu}_{gM,C}$ is smaller than that of $\hat{\mu}_{gM}$, then it is preferable to use $\tilde{\mu}_{gM,C}$ over $\hat{\mu}_{gM}$ to estimate $\mu_g$. For example, \citet{belomestny2020variance} utilize a spectral variance estimator and minimize the estimated asymptotic variance of the control variate estimator, and find improvements over alternative control variate approaches that neglect the auto(cross)-covariance structure of the Markov chain.

In terms of estimating $\Sigma$, given \eqref{eq:avar_sum}, one might consider estimating each $\gamma_g(k) = \textrm{Cov}_\pi (g(X_0), g(X_{k}))$ first, and then using a plug-in estimator $ \sum_{k=-\infty}^{\infty} \hat{\gamma}_g(k)$ based on the estimates $\hat{\gamma}(k)$ to estimate $\Sigma$. A natural estimate for $\gamma_g(k)$ is the empirical covariance matrix $\rinitf{g}(k) \in \R^{d\times d}$, defined as
\begin{align}\label{eq:autocov}
    \rinitf{g}(k) = \begin{cases}
    \frac{1}{M}\sum_{t=0}^{M-1-|k|} \tilde{g}(X_t)\tilde{g}(X_{t+{|k|}})^\top & 0 \le k < M\\  
   \{\frac{1}{M}\sum_{t=0}^{M-1-|k|} \tilde{g}(X_t)\tilde{g}(X_{t+|k|})^\top\}^\top & -M<k<0\\
    0 & |k|\ge M.
    \end{cases} 
\end{align} 
where we define $\tilde{g}(X_t) =g(X_t)-\frac{1}{M}\sum_{t=0}^{M-1} g(X_t)$. Under mild conditions, a Strong Law of Large Numbers (SLLN) holds such that $\rinitf{g}(k)$ converges to $\gamma_g(k)$ for each fixed $k$. However, this convergence does not imply the convergence of the estimator \newline $\sum_{k = -\infty}^\infty \rinitf{g}(k)=\sum_{k=-(M-1)}^{M-1}\rinitf{g}(k)$ to $\Sigma = \sum_{k=-\infty}^\infty \gamma_g(k)$~\citep[e.g.,][]{anderson1971statistical}. As a result, alternative methods for estimating $\Sigma$ have been investigated in previous literature.

In univariate settings where $d=1$, methods for estimating $\Sigma \in \R_{+}$ with better statistical properties have been proposed. These methods include spectral variance estimators~\citep{anderson1971statistical}, estimators based on batch means~\citep{priestley1981spectral,flegal2010batch,chakraborty2022estimating}, and a class of methods for reversible Markov chains including initial sequence estimators~\citep{geyer1992practical} and the moment least squares estimator \citep{berg2022efficient}.
The spectral variance estimators utilize windowed autocovariance sequences as an input, where a window function is often chosen to downweight the contribution from $\rinitf{g}(k)$ for larger $k$ to the asymptotic variance estimate. The batch means and overlapping batch means methods utilize scaled sample variances calculated from the sample means of either non-overlapping or overlapping batches.  The batch means and overlapping batch means estimators turn out to be closely related to the spectral variance estimators~\citep{Damerdji1991-oj, flegal2010batch}. The batch means and spectral variance estimators have known consistency properties. In particular, with an appropriate choice of batch or window size, they are a.s. consistent for $\Sigma$ \citep{flegal2010batch}. 

\citet{geyer1992practical}, on the other hand, proposed initial sequence estimators for estimating $\Sigma$ when $d=1$ by imposing various shape constraints on the estimates of certain sums of autocovariances. More specifically, these estimators exploit positivity, monotonicity, and convexity constraints satisfied for reversible Markov chains by the sequence $\Gamma_g=\{\Gamma_g(k)\}_{k=0}^{\infty}$ defined by $\Gamma_g(k) := \gamma_g(2k)+\gamma_g(2k+1)$ for $k=0,1,2,\dots$. Initial sequence estimators have very strong empirical performance and do not require the choice of a tuning parameter value, making them very useful in practice. For example, the widely used Stan software \citep{stan2019} employs the initial sequence estimators to estimate the effective sample size of Markov chain simulations. However, the consistency of these initial sequence estimators of $\Sigma$ is still unknown. To our knowledge, the only consistency results for the initial sequence estimators are that the initial sequence estimates asymptotically do not underestimate $\Sigma$.

\cite{berg2022efficient} further explored the idea of estimating the autocovariance $\gamma_g$ by enforcing shape constraints and proposed a moment least squares (momentLS) estimator for estimating $\gamma_g$ from a reversible and geometrically ergodic Markov chain. Using the observation that $\gamma_g$ from a reversible Markov chain has a moment representation (see Section~\ref{sec:momentLS_background}) and thus satisfies certain shape constraints, they considered the projection of an input autocovariance sequence estimate onto the set of $[-1+\delta,1-\delta]$ moment sequences, which they called a momentLS sequence estimator. Based on this autocovariance sequence estimator, they also proposed an asymptotic variance estimator to estimate the asymptotic variance of the sample mean of univariate functions of the Markov chain iterates. They showed that both the autocovariance sequence and asymptotic variance estimator are strongly consistent, provided that the hyperparameter $\delta$ is chosen sufficiently small.

Estimating the asymptotic variance matrix in multivariate settings where $d>1$ has also been considered. Many of the estimators can be considered as multivariate generalizations of univariate estimators. For example, multivariate counterparts to the spectral variance estimator and batch means estimator were considered and their theoretical properties were studied in \citet{vats2018strong, vats2019multivariate}. A strong consistency property for the multivariate spectral variance estimator and the strong consistency of multivariate batch means with an appropriate choice of window function and batch size were established in \citet{vats2018strong} and \citet{vats2019multivariate}, respectively.

Various generalizations of the initial sequence estimators in \cite{geyer1992practical} have been proposed in previous literature. For example, \citet{kosorok2000error} constructed an initial positive (matrix) sequence estimator by truncating the cumulative autocovariance matrix at the first time point where the smallest eigenvalue becomes negative. \citet{dai2017multivariate} proposed a different truncation point where they suggested truncating when the generalized variance associated with the current cumulative sum fails to increase. Similar to initial sequence estimators in the univariate setting, no consistency results have been established for the multivariate initial sequence estimators of the asymptotic variance, although they are known to be asymptotically conservative.

In this work, we propose a multivariate sequence estimator for an auto- or cross-covariance sequence associated with a multivariate function $g:\mathsf{X}\to\mathbb{R}^d$ with $d \ge 1$ of a Markov chain based on the work of \cite{berg2022efficient}. We show that our proposed multivariate sequence estimator is strongly consistent and that the resulting multivariate matrix estimator for the asymptotic variance matrix is also strongly consistent, provided an appropriate choice of the hyperparameter $\delta$.
The paper is organized as follows: In Section \ref{sec:momentLS_background}, we describe our problem set-up and provide some background on univariate momentLS estimators. In Section \ref{sec:mtv_momentLS}, we introduce the proposed sequence and asymptotic variance matrix estimators for multivariate functions of Markov chain iterates. The statistical guarantees of the proposed estimators are studied in Section \ref{sec:statistical_guarantees}. In Section \ref{sec:emp}, we compare the proposed estimator with competing methods using simulated and real data examples. 

\section{Problem setup and univariate MomentLS estimator}\label{sec:momentLS_background}

We consider a Markov chain $X=(X_0,X_1,\dots)$, which is a sequence of $\mathsf{X}$-valued random variables. We assume that the state space $\mathsf{X}$ is a complete separable metric space and let $\mathscr{X}$ be the associated Borel $\sigma$-algebra. We let $Q:\mathsf{X}\times\mathscr{X}\to [0,1]$ be the transition kernel for $X$. An initial measure $\nu$ on $\mathscr{X}$ and a transition kernel $Q$ define a Markov chain probability measure $P_\nu$ for $X=(X_0,X_1,X_2,\dots)$ on the canonical sequence space $(\Omega, \mathcal{F})$. We write $E_\nu$ to denote expectation with respect to $P_{\nu}$.

For a probability measure $\pi$ on $(\mathsf{X},\mathscr{X})$, a probability kernel $Q$ is said to be $\pi$-stationary if $\pi (A)=\int Q(x,A)\pi(dx)$ for all $A\in\mathscr{X}$. In addition, we say $Q$ satisfies the reversibility property with respect to $\pi$ if
$\int_A \pi(dx) Q(x, B)=\int_B \pi(dy) Q(y, A)$
for all $A, B \in \mathscr{X}$. We note that reversibility with respect to $\pi$ is a sufficient condition for $\pi$-stationarity of $Q$~\citep[see, e.g.,][Section 6.5.3]{robert2004monte}.
For a function $f:\mathsf{X}\to\mathbb{C}$ and a transition kernel $Q$, we define the Markov operator $Q$ by 
\begin{align}
    Qf(x)=\int Q(x,dy) f(y).\label{eq:Qf}
\end{align}We define $Q^0f(x)=f(x)$, $Q^1f(x)=Q f(x)$, and $Q^tf(x)=Q(Q^{t-1}f)(x)$ for $t>1$, and we define $Q^t(x,A)=Q^t1_A(x)$, where $1_A(\cdot)$ is the indicator function for the set $A$.

We let $L^2(\pi)$ be the space of functions which are square integrable with respect to $\pi$, i.e., $L^2(\pi) = \{f:\mathsf{X} \to \mathbb{C}; \int |f(x)|^2 \pi(dx)< \infty\}$. 
For functions $f, g\in L^2(\pi)$, we define an inner product $\braket{f,g}_\pi = \int f(x) \overline{g(x)} \pi(dx),$ where $\overline{g(x)}$ is the conjugate of $g(x)$.
We note that $L^2(\pi)$ is a Hilbert space equipped with the inner product $\langle\cdot,\cdot\rangle_\pi$. For $f \in L^2(\pi)$, we define $\|f\|_{L^2(\pi)} = \sqrt{\braket{f,f}_\pi}$. Also, for an operator $T$ on $L^2(\pi)$, we define $\vertiii{T}_{L^2(\pi)} = \sup_{f; \|f\|_{L^2(\pi)} \le 1} \|Tf\|_{L^2(\pi)}$ and we say $T$ is bounded if $\vertiii{T}_{L^2(\pi)}<\infty$. We say $T^*$ is the adjoint of $T$ if $\langle Tf, g\rangle_\pi = \langle f,T^*g\rangle_\pi$ for any $f,g \in L^2(\pi)$, and $T$ is self-adjoint if $T=T^*$.

For a Markov kernel $Q$, we note that $Q$ is a linear operator since $Q\{af+bg\}(x) = \int Q(x,dy) \{af+bg\}(y) = a\int Q(x,dy)f(y) + b\int Q(x,dy)g(y)= aQf(x) + bQg(x)$ for any $a,b \in \mathbb{C}$ and $f,g \in L^2(\pi)$. Since $Q$ is a contraction map, $\vertiii{Q}_{L^2(\pi)} \le 1$ and $Q$ is bounded. If the transition kernel $Q$ is reversible with respect to $\pi$, the Markov operator $Q$ is self-adjoint since $\langle Q1_A, 1_B \rangle_\pi = \langle 1_A, Q1_B \rangle_\pi$ for any $A,B \in \mathscr{X}$, and simple functions are dense in $L^2(\pi)$.

The spectrum of the operator $Q$ plays a key role in determining the mixing properties of a Markov chain with transition kernel $Q$. For an operator $T$ on the Hilbert space $L^2(\pi)$, the spectrum of $T$ is defined as $\sigma(T) = \{\lambda \in \mathbb{C}; (T-\lambda I)^{-1} \mbox{does not exist or is unbounded}\}.$ For Markov operators $Q$, we define the spectral gap $\delta(Q)$ of $Q$ as $\delta(Q) = 1- \sup\{ |\lambda| ; \lambda \in \sigma(Q_0) \}$ where $Q_0$ is defined as $Q_0f = Qf - E_\pi[f(X_0)] f_0$ and $f_0\in L^2(\pi)$ is the constant function such that $f_0(x) = 1$ for all $x\in \mathsf{X}$. It is easy to check that $Q_0$ is also self-adjoint when $Q$ is self-adjoint.
If $Q$ is reversible, $Q$ has a positive spectral gap $(\delta(Q)>0)$ if and only if the chain is geometrically ergodic \citep{Roberts1997-qk, kontoyiannis2012geometric}. 

We let the function of interest $g:\mathsf{X}\to\R^d$ be an $\R^d$-valued function such that $g(x) = [g_1(x),\dots,g_d(x)]^\top$ for each $x \in \mathsf{X}$ and each $g_j$ is square integrable with respect to $\pi$. For any real-valued $f_1,f_2:\mathsf{X} \to \R$, we define the cross-covariance sequence $\gamma_{f_1f_2}(k)$
by
\begin{align}\label{def:gamma_g}
\gamma_{f_1f_2}(k)&= {\rm Cov}_\pi\{f_1(X_0), f_2(X_{k})\}  \quad\mbox{for } k\in\mathbb{Z}.
\end{align}
For $f:X\to\mathbb{R}$, $\gamma_{ff}(k)$ is the autocovariance sequence for $\{f(X_t)\}_{t\ge 0}$ and we simply write it as $\gamma_f(k)$. We use $\gamma_{f_1f_2} = \{\gamma_{f_1f_2}(k)\}_{k\in\Z}$ to denote the auto(cross)-covariance sequence on $\Z$.

In the remainder, we consider a discrete time Markov chain $X=\{X_t\}_{t=0}^{\infty}$ with stationary distribution $\pi$ and $\pi$-reversible transition kernel $Q$ with a positive spectral gap. We summarize our assumptions on the Markov chain $X$ as follows for future reference:
\begin{enumerate}[label=(A.\arabic*)]
\setlength\itemsep{0em}
\item \label{cond:harris_ergodicity}(Harris ergodicity) $X$ is $\psi$-irreducible, aperiodic, and positive Harris recurrent.
\item \label{cond:piReversible}(Reversibility) $Q$ is $\pi$-reversible for a probability measure $\pi$ on $(\mathsf{X},\mathscr{X})$.
\item \label{cond:geometric_ergodicity}(Geometric ergodicity) There exists a real number $\rho<1$ and a non-negative function $M$ on the state space $\mathsf{X}$ such that
    $\|Q^n(x,\cdot) - \pi(\cdot)\|_{\rm TV} \le M(x)\rho^n, \mbox{  for all } x\in \mathsf{X},$
where $\|\cdot\|_{\rm TV}$ is the total variation norm.
\end{enumerate}
Throughout the paper, we assume that each component of $g:\mathsf{X}\to\mathbb{R}^d$ is in $L^2(\pi)$, i.e,
\begin{enumerate}[label=(B.\arabic*)]
	\item \label{cond:integrability} (Square integrability)$\int g_j(x)^2\pi(dx)<\infty$, for $j=1,\dots,d$.
\end{enumerate}
For the definitions of $\psi$-irreducibility, aperiodicity, and positive Harris recurrence, see e.g., \citet{meyn2009markov}. Reversibility and geometric ergodicity are the two key assumptions for the class of Markov chains we consider in this paper. While reversibility is not necessary to obtain a $\pi$-stationary chain, many popular Markov chains possess this property, including Metropolis-Hastings and random scan Gibbs samplers. Geometric ergodicity implies exponential convergence of the Markov chain $X$ to its target distribution $\pi$. When the state space $\mathsf{X}$ is finite, all irreducible and aperiodic Markov chains are geometrically ergodic. While geometric ergodicity is not guaranteed for infinite state spaces, geometric ergodicity is still of practical and theoretical importance. For instance, geometric ergodicity is desirable for Markov chain samplers because non-geometrically ergodic chains often exhibit slow convergence \citep{roberts1998markov}. In addition, geometric ergodicity provides one of the simplest sufficient conditions under which a Markov chain CLT holds. Specifically, for a reversible and geometrically ergodic chain, the existence of a finite second moment of the function of interest $g$ is sufficient to establish a CLT
\begin{align}\label{eq:clt}
    \sqrt{M}(\hat{\mu}_{gM} - \mu_g) \overset{d}{\to} N(0, \Sigma)
\end{align} as $M\to\infty$, 
where $\hat{\mu}_{gM} = \frac{1}{M}\sum_{t=0}^{M-1} g(X_t)$ and $\mu_g = \int g(x) \pi(dx)$~\citep[e.g.,][]{jones2004central}.

\citet{berg2022efficient} proposed shape-constrained estimators for both autocovariance sequence and asymptotic variance estimation for the univariate case $g:\mathsf{X}\to\mathbb{R}^d$ for $d=1$. These estimators are based on the fact that the autocovariance sequence $\gamma_{g}$ from a reversible Markov chain is a $[-1,1]$-\textit{moment} sequence, i.e., there exists a positive measure $F_{g}$ supported on $[-1,1]$ such that 
\begin{align}
    \gamma_g(k)=\int \alpha^{|k|}\,F_g(d\alpha),\;\;\;\;\;\;k\in \mathbb{Z}.\label{eq:moments}
\end{align}
In other words, each $\gamma_g(k)$ is the $|k|$th moment of $F_g$. The moment representation~\eqref{eq:moments} of $\gamma_g$ imposes specific shape restrictions on $\gamma_g$, and to exploit this, \citet{berg2022efficient} proposed to project a point-wise consistent initial autocovariance sequence estimator $\rinitf{g}$ onto the set of moment sequences satisfying a representation as in~\eqref{eq:moments}. 

To be more concrete, let us first introduce some definitions and notation. For $a,b \in\R$ such that $-\infty< a \le b < \infty$, we say that a sequence $m$ is an \textit{$[a,b]$-moment sequence} if there exists a positive measure $\mu$ supported on $[a,b]$ such that the equation
\begin{align}\label{eq:momentRep}
    m(k) = \int x^{|k|} \mu(dx)
\end{align}
holds for any $k \in \Z$. The positive measure $\mu$ is called a \textit{representing measure} for the sequence $m$. We also define the moment space $\mathscr{M}_\infty ([a,b])$ as the set of $[a,b]$-moment sequences. When $b>0$ and $a = -b$, we simply write $\mathscr{M}_\infty (b)= \mathscr{M}_\infty ([-b,b])$. We use $\ell_2(\Z)$ to denote the set of square-summable real-valued sequences on $\Z$, i.e., $\ell_2(\Z) = \{m \in \R^\Z; \|m\|_2^2 <\infty\}$ where we denote $\|m\|_2^2 = \sum_{k\in\Z} m(k)^2$.

The moment least squares sequence (momentLS) estimator $\Pi_\delta(\rinitf{g})$ resulting from an initial autocovariance sequence estimator $\rinitf{g}\in\ell_2(\mathbb{Z})$ is defined as the projection of $\rinitf{g}$ onto the $[-1+\delta,1-\delta]$ moment sequence space, i.e.,
 \begin{align}\label{eq:momentLS}
    \Pi_\delta(\rinitf{g}) =\underset{m\in \Mld{\delta}}{\arg\min}\sum_{k\in\mathbb{Z}}\{\rinitf{g}(k)-m(k)\}^2.
\end{align}
Here $\delta>0$ is a hyperparameter for the momentLS estimator $\Pi_\delta(\rinitf{g})$. We require the hyperparameter $\delta$ to be chosen sufficiently small such that the true autocovariance $\gamma_g$ is a feasible point for the optimization problem 
\eqref{eq:momentLS}. That is, we require ${\rm Supp}(F_g) \subseteq [-1+\delta, 1-\delta]$, where $F_g$ is the representing measure for $\gamma_g$. For a measure $F$ on supported on $[-1,1]$, with the definition 
\begin{align}\label{def:delta_gamma}
    \Delta(F) = 1-\max\{|x|; x\in {\rm Supp}(F)\},
\end{align} the condition ${\rm Supp}(F_g) \subseteq [-1+\delta, 1-\delta]$ on $\delta$ is equivalent to $0<\delta \le \Delta(F_g)$. For a reversible, geometrically ergodic Markov chain, $\Delta(F_g) \ge \delta(Q) >0$, and therefore such a choice of $\delta$ exists.
Once the autocovariance sequence $\gamma_g$ is estimated by $\Pi_\delta(\rinitf{g})$, \citet{berg2022efficient} proposed to estimate the asymptotic variance $\Sigma = \sum_{k\in\Z} \gamma_g(k)$ in \eqref{eq:clt} by $\hat{\Sigma} = \sum_{k\in \Z} \Pi_\delta(\rinitf{g})(k)$. 
Although the definition of this estimator involves an infinite sum, it is still computable. As $\Pi_\delta(r_{g,M})$ is the projection onto the $[-1+\delta,1-\delta]$ moment sequence space, for each $k \in \Z$, $\Pi_\delta(r_{g,M})(k) = \int \alpha^{|k|}\hat{\mu}_\delta (d\alpha)$ for a measure supported on $[-1+\delta,1-\delta]$. Moreover, it is shown that the measure $\hat{\mu}_{\delta}$ is discrete with at most $M+1$ points of support~\citep{berg2022efficient}. Thus, $\hat{\Sigma} = \sum_{k\in \Z} \Pi_\delta(\rinitf{g})(k)$ can be computed by 
\begin{align}
    \sum_{k\in\mathbb{Z}} \int \alpha^{|k|}\,\hat{\mu}_{\delta}(d\alpha)=\int \sum_{k\in\mathbb{Z}}\alpha^{|k|}\,\hat{\mu}_{\delta}(d\alpha)=\int \frac{1+\alpha}{1-\alpha}\,\hat{\mu}_{\delta}(d\alpha)=\sum_{\alpha:\hat{\mu}_{\delta}(\{\alpha\})>0}\,\frac{1+\alpha}{1-\alpha}\hat{\mu}_{\delta}(\alpha)\label{eq:infiniteSum}
\end{align} where the last expression involves a sum over the finite, discrete support of $\hat{\mu}_{\delta}$.

\citet{berg2022efficient} showed for univariate $g:\mathsf{X}\to\mathbb{R}$ that if a valid initial autocovariance sequence estimator $\rinitf{g}$ is used, in the sense that the initial estimator $\rinitf{g}$ satisfies the following conditions
\begin{enumerate}[label=(R.\arabic*)]
\setlength\itemsep{0em}
    \item \label{cond:R1}(a.s. elementwise convergence) $\rinitf{g}(k) \underset{M\to\infty}{\to} \gamma_g(k)$ for each $k\in \Z$, $P_x$-almost surely, for any initial condition $x \in \mathsf{X}$,
    \item \label{cond:R2}(finite support) $\rinitf{g}(k) =0 $ for $k\ge n(M)$ for some $n(M)<\infty$, and
    \item \label{cond:R3}(even function with a peak at 0) $\rinitf{g}(k) = \rinitf{g}(-k)$ and $\rinitf{g}(0) \ge |\rinitf{g} (k)|$ for each $k\in \Z$,
\end{enumerate}
and a valid $\delta$ is used, in the sense that $0<\delta \le \Delta(F_g)$, then the momentLS sequence and asymptotic variance estimators are strongly consistent for $\gamma_g$ and $\Sigma$, respectively. For future reference, we summarize the consistency results for univariate momentLS estimators as follows:
\begin{thm}\label{thm:univariate_conv} Suppose $X_{0},X_1,...,$ is a Markov chain with transition kernel $Q$ satisfying~~\ref{cond:harris_ergodicity}-\ref{cond:geometric_ergodicity}, and suppose $g:\mathsf{X}\to\mathbb{R}$ satisfies~\ref{cond:integrability}. Let $\gamma_g$ denote the autocovariance sequence as defined in \eqref{def:gamma_g}, and let $F_g$ denote the representing measure for $\gamma_g$. Suppose $\delta$ is chosen so that $0<\delta \le \Delta(F_g)$. Let $\rinitf{g}$ be an initial autocovariance sequence estimator for $\gamma_g$ satisfying conditions \ref{cond:R1} - \ref{cond:R3}. 
Then for each initial condition $x\in \mathsf{X}$, 
\begin{enumerate}
\setlength\itemsep{0em}
    \item $\sum_{k \in \Z} \{\gamma_g(k)-\Pi_\delta(\rinitf{g})(k)\}^2 \underset{M\to\infty}{\to} 0, \,\,\,P_x\mbox{-a.s.}$
    \item $P_x(\m_{g,M} \to F_g \mbox{ vaguely, as }M\to\infty)=1$, where $\m_{g,M}$ and $F_g$ are the representing measures for $\Pi_\delta(\rinitf{g})$ and $\gamma_g$, and
    \item $\sigma^2(\Pi_\delta(\rinitf{g})) \to \sigma^2(\gamma_g)\,\,\,P_x\mbox{-a.s.}$
\end{enumerate}
where we define $\sigma^2(m) = \sum_{k\in\Z}m(k)$ for a sequence $m$ on $\Z$, and we recall the definition of vague convergence of measures as follows: a sequence of measures $\{\nu_n\}_{n\in \mathbb{N}}$ and $\nu$ on $\R$, $\nu_n$ is said to \textit{converge vaguely} to $\nu$ if and only if $\int f d\nu_n \to \int f d\nu$ for all  $f \in C_0(\R)$ [e.g., \citealp{folland1999real}], where $C_0(\R)$ is the space of continuous functions that vanish at infinity.
\end{thm}

\section{Multivariate MomentLS estimator}\label{sec:mtv_momentLS}
We now introduce the multivariate momentLS estimator for the cross-covariance sequence $\gamma_{g_i,g_j}$ and asymptotic variance $\Sigma$ in the Markov chain central limit theorem \eqref{eq:clt}. Recall that for a Markov chain $X=(X_0,X_1,\dots)$ and an $\R^d$-valued function $g:\mathsf{X}\to \R^d$, we defined the cross (auto)-covariance sequences $\gamma_{g_i,g_j}$ for $i,j= 1,\dots,d$ as
 \begin{align*}
     \gamma_{g_i,g_j}(k) = {\rm Cov}_\pi(g_i(X_0), g_j(X_k))
 \end{align*}
 for $i=1,\dots,d$ and $k\in \Z$, where we define ${\rm Cov}_\pi\{g_i(X_0), g_j(X_{k})\} = {\rm Cov}_\pi\{g_i(X_{|k|}), g_j(X_{0})\}$ for $k<0$, and we write $\gamma_{g_i} =  \gamma_{g_i,g_i}$.

We first note that $\gammaij(k) = \gammaji(-k)$ for $k \in \Z$ by definition, as for $k\in\N$, $\gammaji(-k) = {\rm Cov}_\pi(g_j(X_{|k|}), g_i(X_0)) ={\rm Cov}_\pi( g_i(X_0),g_j(X_{k}),) = \gammaij(k)$. In addition, for a reversible Markov chain, $\gammaij(k) = \gammaji(k)$ for $k \in \Z$. This is because for $k \ge 0$, we have from self-adjointness of $Q$ that
    $\gammaij(k)= {\rm Cov}_\pi(g_i(X_0), g_j(X_k))  = \langle \check{g}_{i}, Q^k \check{g}_{j}\rangle_\pi =\langle Q^k\check{g}_{i},  \check{g}_{j}\rangle_\pi = {\rm Cov}_\pi(g_i(X_k), g_j(X_0)) = \gammaji(k),$
where we define the true-mean centered version $\check{g}_i$ of $g_i$ by $\check{g}_i(x)=g_i(x)-E_\pi[g_i(X_0)]$. For $k <0$, the result follows from $\gammaij(k) = \gammaji(-k)$. 

When $i=j$, i.e., $\gamma_{g_i,g_i}$ is an autocovariance sequence, $\gamma_{g_i,g_i}$ admits the following moment representation, 
\[
\gamma_{g_i,g_i}(k) = \int\lambda^{|k|} F_{g_i,g_i}(d\lambda)
\]
for any $k\in \Z$, where $F_{g_i,g_i}$ is a positive measure with ${\rm Supp}(F_{g_i,g_i}) \subseteq [-1+\delta(Q), 1-\delta(Q)]$, where we recall that $\delta(Q)$ is the spectral gap of $Q$~\citep[see, e.g.,][]{geyer1992practical}. In other words, $\gamma_{g_i,g_i}$ is a $[-1+\delta(Q), 1-\delta(Q)]$ moment sequence. 
In fact, $\gammaij$ when $i\ne j$ admits a similar mixture representation, although the corresponding representing measure $F_{g_i,g_j}$ may not necessarily be positive. Moreover, all of these (signed) measures, denoted as $F_{g_i,g_j}$, arise from a \textit{spectral resolution} of the Markov kernel $Q$. A notable implication is that all $F_{g_i,g_j}$ are supported on the same interval $[-1+\delta(Q), 1-\delta(Q)]$.

To be more concrete, let us first define a spectral resolution for a Hilbert space $\mathcal{H}$ equipped with an inner product $\langle \cdot, \cdot \rangle_{\mathcal{H}}$. Let $\|\cdot\|_{\mathcal{H}}$ be the norm induced by the inner product. A function $E: \R \to \mathscr{B}(\mathcal{H})$ is called a spectral resolution on $\mathcal{H}$ if the function $E$ satisfies the following \citep[see e.g., Chapter 6 in][]{stein2009real}: 1. $\|E(\lambda) f\|_{\mathcal{H}}$ is an increasing function of $\lambda$ for every $f \in \mathcal{H}$, 2. there is an interval $[a, b]$ such that $E(\lambda)=0$ if $\lambda<a$, and $E(\lambda)=I$ if $\lambda \geq b$ where $I$ denotes the identity operator on $\mathcal{H}$, and 3. $\lim _{\mu \downarrow \lambda} E(\mu) f=E(\lambda) f $ for any $f \in \mathcal{H}$, where $\mathscr{B}(\mathcal{H})$ denotes the space of bounded linear operators on $\mathcal{H}$. It can be shown that the property 1 is equivalent to the condition that for any $\lambda<\mu$, $E(\mu)-E(\lambda)$ is an orthogonal projection operator.

For a reversible chain with a Markov kernel $Q$, by the spectral representation theorem for self-adjoint bounded operators in Hilbert space \citep[e.g., Theorem 6.1 in][]{stein2009real}, we can obtain a spectral resolution $E$ such that for any $f,g \in L_2(\pi)$ 
\begin{align}\label{eq:spect_resolution}
    \langle Q_0 f,g \rangle_\pi = \int \alpha\, \mE{f}{g} (d\alpha)
\end{align}
where we denote by $\mE{f}{g}$ the complex-valued measure on $\mathscr{B}_\R$ uniquely defined by the complex-valued function $\Psi_{f,g}:\alpha \to \langle E(\alpha)f,g\rangle_\pi$ through $\mE{f}{g}((-\infty,\alpha]) = \Psi_{f,g}(\alpha)$. When $f=g$, we let $F_{f} = F_{f,f}$ for notational simplicity. The following Proposition \ref{prop:rep_gamma_ij} summarizes the properties of the mixing measure $F_{g_i,g_j}$ for the cross-covariance sequence $\gammaij$.

\begin{prop}\label{prop:rep_gamma_ij}
Suppose $X$ satisfies \ref{cond:piReversible} and $g:\mathsf{X} \to \R^d$ satisfies \ref{cond:integrability}. There exists a finite signed measure $F_{g_i,g_j}$ on $(\R, \mathscr{B}_\R)$ with ${\rm Supp}(F_{g_i,g_j}) \subseteq [-1+\delta(Q), 1-\delta(Q)]$ such that 
\begin{align}\label{prop:eq:gamma_ij}
    \gammaij(k) = \int\lambda^{|k|} \, F_{g_i,g_j}(d\lambda)
\end{align}
for each $k \in \Z$ and $i,j \in \{1,\dots,d\}$.
Moreover, the mixing measure $F_{g_i,g_j}$ in \eqref{prop:eq:gamma_ij} is unique, and is given by $F_{g_i,g_j}((-\infty,\alpha]) = \langle E(\alpha)g_i, g_j\rangle_\pi$ for $\alpha \in \R$ where $E$ is the spectral resolution in \eqref{eq:spect_resolution}.
Finally, when $i=j$, the measure $F_{g_i} = F_{g_i,g_i}$ is positive. 
\end{prop}
The proof is deferred to Supplementary Material\ifshow~\ref{sec:proof_of_prop1}.\else~S1.1. \fi
We note in general $\gammaij(k)$ is not a moment sequence, as $F_{g_i,g_j}$ is not necessarily a positive measure when $i\ne j$. However, we can represent $F_{g_i,g_j}$ as a difference of positive measures, and consequently represent $\gammaij(k)$ as a difference of moment sequences.

Note that for any $f \in \mathcal{H}$, $\Psi_{f,f}: \alpha \to \langle E(\alpha)f,f\rangle_\pi$ is a positive, increasing function of $\alpha$, as $\Psi_{f,f}(\alpha) = \|E(\alpha)f\|_{L^2(\pi)}^2$ since each $E(\alpha)$ is a projection operator and $\alpha\to \|E(\alpha)f\|_{L^2(\pi)}$ is an increasing function of $\alpha$. By the polarization identity and using the fact that $\langle E(\alpha) g_i, g_j\rangle_\pi$ is real due to Proposition \ref{prop:rep_gamma_ij}, for any $a,b \ne 0$,
\begin{align}\label{eq:polarization_identity}
    \Psi_{g_i,g_j}(\alpha) 
    &= \langle E(\alpha)g_i, E(\alpha)g_j\rangle_\pi \nonumber\\
    &=\frac{1}{4ab}\{\|E(\alpha)(ag_i + bg_j)\|_\pi^2 - \|E(\alpha)(ag_i - bg_j)\|_\pi^2\} \nonumber\\
    &=\frac{1}{4ab}\{\Psi_{ag_i+bg_j,ag_i+bg_j}(\alpha) - \Psi_{ag_i-bg_j,ag_i - bg_j}(\alpha)\},
\end{align}
and therefore, the measure $F_{g_i,g_j}$ can be written as the difference of two positive measures $F_{g_i,g_j}=\frac{1}{4ab}\{F_{ag_i+bg_j}-F_{ag_i-bg_j}\}$ where $F_{ag_i+bg_j}$ and $F_{ag_i - bg_j}$ are positive measures uniquely defined by positive, increasing functions $\Psi_{ag_i+bg_j,ag_i+bg_j}$ and $\Psi_{ag_i-bg_j,ag_i - bg_j}$ respectively. Consequently, $\gammaij = \{\gammaij(k)\}_{k \in \Z}$ can also be represented as the difference of two moment sequences, since
\begin{align}
    \gammaij(k) 
    &=\frac{1}{4ab} \int \alpha^{|k|} F_{ag_i+bg_j}(d\alpha)-\frac{1}{4ab} \int \alpha^{|k|} F_{ag_i-bg_j}(d\alpha)\label{eq:gamma_ij_mseq}
\end{align}
and both $\gamma_{ag_i+bg_j}(k) = \int \alpha^{|k|} F_{ag_i+bg_j}(d\alpha)$ and $\gamma_{ag_i+bg_j}(k) = \int \alpha^{|k|} F_{ag_i-bg_j}(d\alpha)$ are moment sequences belonging to $\gamma_{ag_i + bg_j} \in \Mld{\Delta(F_{ag_i + bg_j})}$ and $\gamma_{ag_i - bg_j} \in \Mld{\Delta(F_{ag_i - bg_j})}$.

In light of this observation, we can estimate the $\gammaij$ sequence by employing the momentLS sequence estimators for the autocovariance sequences of the two univariate processes $\{a g_i (X_t) + bg_j (X_t)\}_{t\ge 0}$ and $\{a g_i (X_t) - bg_j (X_t)\}_{t\ge 0}$. Note that the choice of $a$ and $b$ determines the weighting applied to $g_i$ and $g_j$ when constructing these two univariate processes. By Theorem \ref{thm:univariate_conv}, as long as we use $\delta$ such that $0< \delta \le \min \{ \Delta(F_{ag_i + bg_j}), \Delta(F_{ag_i - bg_j})\}$, both $\Pi_\delta (\rinitf{ag_i +bg_j})$ and $\Pi_\delta (\rinitf{ag_i -bg_j})$ will be strongly consistent for $\gamma_{ag_i +bg_j}$ and $\gamma_{ag_i - bg_j}$, and therefore $\hat{\gamma}_{M,ij} = (4ab)^{-1} \{\Pi_\delta (\rinitf{ag_i +bg_j}) -\Pi_\delta (\rinitf{ag_i - bg_j})\}$ will be strongly consistent for $\gammaij$. Furthermore, the following Proposition shows that $\min \{ \Delta(F_{ag_i + bg_j}), \Delta(F_{ag_i - bg_j})\}$ can be further lower-bounded by $\min\{\Delta(F_{g_i}),\Delta(F_{g_j})\}$, therefore the permissible range of values for $\delta$ for estimating $\gamma_{ag_i+ bg_j}$ and $\gamma_{ag_i - bg_j}$ can be given based on $\Delta(F_{g_i})$ and $\Delta (F_{g_j})$ for any values of $a$ and $b$.

\begin{prop}\label{prop:delta_bound}
Let $g:\mathsf{X} \to \R^d$ be an $\R^d$-valued function such that each $g_j \in L^2(\pi)$, $j=1,\dots,d$. For any $\bc \in \R^d$, 
    ${\rm Supp} (F_{\bc^\top g} ) \subseteq \cup_{i; \bc_i \ne 0} {\rm Supp}(F_{g_i})$
\end{prop}
The proof is deferred to Supplementary Material\ifshow~\ref{sec:proof_of_prop2}. \else~S1.2. \fi
Note that if $\delta_i$ and $\delta_j$ are valid in the sense that $0<\delta_i \le \Delta(F_{g_i})$ and $0<\delta_j \le \Delta(F_{g_j})$, then $\delta_{ij} = \min\{\delta_i,\delta_j\}$ is also a valid input for estimating $\gamma_{ag_i + b g_j}$ and $\gamma_{ag_i - b g_j}$. This is because by Proposition \ref{prop:delta_bound},
\begin{align*}
\operatorname{Supp}(F_{ag_i + bg_j}) &\subseteq \operatorname{Supp}(F_{g_i}) \cup \operatorname{Supp}(F_{g_j})  \\&\subseteq [-1+\min\{\Delta(F_{g_i}),\Delta(F_{g_j})\},1-\min\{\Delta(F_{g_i}),\Delta(F_{g_j})\}]    \\
&\subseteq [-1+\delta_{ij}, 1-\delta_{ij}].
\end{align*}
and similarly $\operatorname{Supp}(F_{ag_i - bg_j}) \subseteq [-1+\delta_{ij}, 1-\delta_{ij}]$ as well.

We now propose the following momentLS sequence estimator  for $\gammaij$ on $\Z$. For a given vector $\boldsymbol{\delta} = [\delta_1, \dots, \delta_d]$ and a weight vector $\mathbf{w}_M = [w_{M1},\dots,w_{Md}]$, we define the momentLS estimator for $\gammaij$, $i \le j$, as
\begin{align}
    &\hat{\gamma}^{(\boldsymbol{\delta})}_{M,ij}(\cdot; \mathbf{w}_M) =\begin{cases}
    \Pi_{\delta_{i}} (r_{M,g_i})(\cdot) & i=j\\
        \frac{1}{4w_{Mi}w_{Mj}} \{ \Pi_{\delta_{ij}}(r_{M,w_{Mi} g_i+w_{Mj} g_j})(\cdot) - \Pi_{\delta_{ij}}(r_{M,w_{Mi} g_i-w_{Mj} g_j})(\cdot) \} & i< j
    \end{cases} \label{eq:gammaHatDef}
\end{align}
where $\delta_{ij} = \min\{\delta_i, \delta_j\}$. When $i> j$, we simply let  $\hat{\gamma}^{(\boldsymbol{\delta})}_{M,ij}(\cdot; \mathbf{w}_M) = \hat{\gamma}^{(\boldsymbol{\delta})}_{M,ji}(\cdot; \mathbf{w}_M)$.

The weight vector $\mathbf{w}_M$ determines relative weights for each $g_i$. We allow $\mathbf{w}_M$ to be random as long as each element $w_{Mi}$ converges to a non-zero number $w_i \ne 0$ almost surely. In the next section, we will show that the momentLS estimator with random weights $\mathbf{w}_M$ will converge to the same quantity as the momentLS estimator with non-random weights $\mathbf{w}$ if $\mathbf{w}_M$ converges to $\mathbf{w}$ elementwise almost surely. In terms of the actual choice of weights, we propose to scale each $g_i$ based on their size, specifically using the norm of each centered $g_i$, i.e., $\|g_i-E_\pi[g_i(X_0)]\|_{L^2(\pi)} = \{\int (g_i(x)-E_\pi[g_i(X_0)])^2 \pi(dx)\}^{1/2} = \sqrt{\gamma_{g_i}(0)}$. The lag-0 autocovariance $\gamma_{g_i}(0)$ is the variance of $g_i(X)$ for $X\sim\pi$, and thus $\sqrt{\gamma_{g_i}(0)}>0$ provided $g_i$ is not $\pi$-a.e. constant. We propose to use $\mathbf{w}_{M}=\mathbf{\tilde{w}}_{M}$ where the $i$th entry $\tilde{w}_{Mi}$ of $\mathbf{\tilde{w}}_{M}$ is the inverse of the estimate of $\sqrt{\gamma_{g_i}(0)}$ based on the lag-0 empirical autocovariance, i.e., 
\begin{align}
\tilde{w}_{Mi} =1/\sqrt{\rinitf{g_i}(0) },\,\,\,i=1,\dots,d. \label{def:w_M}    
\end{align} 
When we use the choice $\mathbf{w}_{M}=\mathbf{\tilde{w}}_{M}$, we will drop the dependence on $\mathbf{w}_M$ in $\hat{\gamma}^{(\boldsymbol{\delta})}_{M,ij}(\cdot; \mathbf{w}_M)$ and define $ \hat{\gamma}^{(\boldsymbol{\delta})}_{M,ij}(\cdot) = \hat{\gamma}^{(\boldsymbol{\delta})}_{M,ij}(\cdot; \mathbf{\tilde{w}}_M)$. Note for $\tilde{w}_{Mi}$ in~\eqref{def:w_M}, we have $\tilde{w}_{Mi} \to 1/\sqrt{\gamma_{g_i}(0)}$ almost surely as $M\to \infty$ for $i=1,\dots,d$.

We also define the momentLS estimator for the asymptotic variance of $\Sigma$ in the multivariate Markov chain CLT.
We first propose a plug-in estimator that essentially estimates each $\Sigma_{ij}$ element-wise. We define $\hat{\Sigma}_{pw}^{(\boldsymbol{\delta})}\in\mathbb{R}^{d\times d}$ by $\hat{\Sigma}_{pw,ij}^{(\boldsymbol{\delta})} = \sum_{k\in\mathbb{Z}} \hat{\gamma}^{(\boldsymbol{\delta})}_{M,ij}(k;\mathbf{\tilde{w}}_M)$ based on the momentLS estimates of $\gammaij$, for each $1\le i, j \le d$. More concretely, from the definition of $\hat{\gamma}^{(\boldsymbol{\delta})}_{M,ij}(\cdot; \mathbf{w}_M) $ in~\eqref{eq:gammaHatDef} and $\mathbf{\tilde{w}}_M$, we have \begin{align}
\hat{\Sigma}_{pw,ij}^{(\boldsymbol{\delta})}=
\begin{cases}
\sum_{k\in \Z}\Pi_{\delta_{i}} (r_{M,g_i})(k) & i=j\\
\frac{s_is_j}{4} \{ \sum_{k\in\Z}\Pi_{\delta_{ij}}(r_{M,g_i/s_i+g_j/s_j})(k) - \sum_{k\in\Z}\Pi_{\delta_{ij}}(r_{M,g_i/s_i-g_j/s_j})(k)\} & i< j
\end{cases}\label{eq:Sigma_pwDefII}
\end{align} and $\hat{\Sigma}_{pw,ij}:=\hat{\Sigma}_{pw,ji}$ for $i>j$. We note that even though the definition of $\hat{\Sigma}_{pw,ij}^{(\boldsymbol{\delta})}$ in \eqref{eq:Sigma_pwDefII} involves infinite sums, these are still computable quantities, as discussed in \eqref{eq:infiniteSum}.

One drawback of the proposed estimator is that the resulting asymptotic variance estimator $\hat{\Sigma}^{(\boldsymbol{\delta})}_{pw}$ is not necessarily positive semi-definite. In such cases, we propose a principled refinement of the first estimator which results in a valid covariance matrix. The key idea is to re-estimate the eigenvalues of the asymptotic variance matrix using momentLS estimators based on the eigenvectors from the first estimator $\hat{\Sigma}^{(\boldsymbol{\delta})}_{pw}$. To illustrate the idea, suppose $\lambda_1\ge \lambda_2\ge \dots \ge \lambda_d$ are eigenvalues of $\Sigma$ and let $\U = [\U_1,\dots,\U_d]$ where $U_1,...,U_d$ denote the corresponding orthogonal eigenvectors. For expositional simplicity, suppose for now that all eigenvalues are distinct. If we choose $\bc = \U_j$, the asymptotic variance $\sigma^2(\gamma_{\bc^\top g})$ will be $\lambda_j$, since $\sigma^2(\gamma_{\bc^\top g}) = \sum_{k\in \Z} \bc^\top \operatorname{Cov}_\pi (g(X_0), g(X_k)) \bc = \bc^\top \Sigma \bc = \lambda_j$.  
Then, for each $j=1,\dots,d$, if we have a consistent (up to sign) estimator $\hat{\U}_{Mj}$ for $\U_j$, we expect that we can consistently estimate the $j$th eigenvalue $\lambda_j$ using the momentLS estimator $\Pi_{\delta}(\rinitf{\hat{\U}_{Mj}^\top g})$, provided we choose a sufficiently small $\delta>0$. Moreover, considering Proposition \ref{prop:delta_bound}, $\delta = \min\{\delta_1,\dots,\delta_d\}$ will be a valid choice of $\delta$ if each $\delta_i$ is chosen to be valid for estimating $\gamma_{g_i}$ for $i=1,\dots,d$. 

Now we introduce a positive semi-definite refinement estimator $\hat{\Sigma}^{(\boldsymbol{\delta})}_{psd}$ from the plug-in estimator $\hat{\Sigma}^{(\boldsymbol{\delta})}_{pw}$ and the momentLS asymptotic variance estimator $\hat{\Sigma}^{(\boldsymbol{\delta})}$ as a combination of $\hat{\Sigma}^{(\boldsymbol{\delta})}_{pw}$ and $\hat{\Sigma}^{(\boldsymbol{\delta})}_{psd}$. 
Since $\hat{\Sigma}^{(\boldsymbol{\delta})}_{pw}$ is a symmetric matrix by construction, there exist real eigenvalues $\hat{\lambda}_1^{pw}\ge \dots \ge \hat{\lambda}_d^{pw}$ and eigenvectors $[\hat{\U}_{M1},\dots,\hat{\U}_{Md}]$ such that $\hat{\Sigma}^{(\boldsymbol{\delta})}_{pw} \hat{\U}_{Mj} = \hat{\lambda}_j^{pw} \hat{\U}_{Mj}$, for $j=1,\dots,d$. We then re-estimate each $\lambda_j$ by $\sigma^2(\Pi_{
\underline{\delta}}(\rinitf{\hat{\U}_{Mj}^\top g}))$, the momentLS asymptotic variance estimator for $\{\hat{\U}_M^\top g(X_t)\}_{t \ge 0}$, where $\underline{\delta} = \min_{1 \le i \le d} \delta_i$. We define $\hat{\Sigma}^{(\boldsymbol{\delta})}_{psd}$ by 
\begin{align}\label{def:multi-LSE2}
    \hat{\Sigma}^{(\boldsymbol{\delta})}_{psd}= \hat{\mathbf{U}}_M \hat{\bLambda}_{\boldsymbol{\delta}} \hat{\mathbf{U}}_M^\top
\end{align}
where $\hat{\mathbf{U}}_M=[\hat{\U}_{M1},\dots,\hat{\U}_{Md}] \in \R^{d\times d}$ and $\hat{\bLambda}_{\boldsymbol{\delta}}$ is a $d$ by $d$ diagonal matrix with $j$th diagonal entry $\sigma^2(\Pi_{
\underline{\delta}}(\rinitf{\hat{\U}_{Mj}^\top g}))$.
Finally, we define our momentLS estimator $\hat{\Sigma}^{(\boldsymbol{\delta})}$ for the asymptotic variance $\Sigma$ as
\begin{equation}
    \hat{\Sigma}^{(\boldsymbol{\delta})} = 
    \begin{cases}
        \hat{\Sigma}^{(\boldsymbol{\delta})}_{pw} & \lambda_{\rm min} (\hat{\Sigma}_{pw}^{(\boldsymbol{\delta})}) \ge 0 \\
        \hat{\Sigma}^{(\boldsymbol{\delta})}_{psd} &  \lambda_{\rm min} (\hat{\Sigma}_{pw}^{(\boldsymbol{\delta})}) < 0.
    \end{cases}\label{eq:momentLSavar}
    \end{equation}

In the next section, we will show that both the plug-in estimator and the refined positive semi-definite matrix estimator are strongly consistent for $\Sigma$, provided that the input vector $\boldsymbol{\delta}$ is valid in the sense that $0<\delta_i \le \Delta(F_{g_i})$ for all $i=1,\dots,d$. This ensures the strong consistency of the moment LS estimator $\hat{\Sigma}^{(\boldsymbol{\delta})}$ for $\Sigma$.

\section{Statistical guarantees}\label{sec:statistical_guarantees}
First, we have the following Lemma which shows that the momentLS sequence and asymptotic variance estimators based on \textit{random} weights $\bc_M$ are still strongly consistent for $\gamma_{\bc^\top g}$, if $\bc = \lim_{M\to \infty} \bc_M$ $P_x$-almost surely for any initial condition $x \in \mathsf{X}$. 

\begin{lem}\label{lem:stochastic_g_conv}
    Let $\mathbf{c}_M = [c_{M1},\dots,c_{Md}]$ denote a sequence of random vectors satisfying $\lim_{M\to \infty}\mathbf{c}_M = \bc$, $P_x$-almost surely for any initial condition $x \in \mathsf{X}$.
    Suppose the Markov chain $X$ satisfies \ref{cond:harris_ergodicity} - \ref{cond:geometric_ergodicity}. Let $g:\mathsf{X} \to \R^d$ be a function satisfying \ref{cond:integrability}.
    Suppose $\delta$ is chosen so that $0<\delta\le\Delta(F_{\bc^\top g})$.
    We have
        1. $\|\Pi_\delta(\rinitf{\mathbf{c}_M^\top g}) - \gamma_{\bc^\top g}\|_2 \to 0$, and
        2. $|\sigma^2(\Pi_\delta(\rinitf{\mathbf{c}_M^\top g}))- \sigma^2(\gamma_{\bc^\top g})|\to 0$, 
    $P_x$-almost surely, for any initial condition $x \in \mathsf{X}$.
\end{lem}

The proof of Lemma~\ref{lem:stochastic_g_conv} is in Supplementary Material\ifshow~\ref{sec:proof_of_lemma1}. \else~S2.1. \fi Now we present the following Theorem which guarantees the almost sure convergence of the plug-in sequence and asymptotic variance estimators. 

\begin{thm}\label{thm:est1-conv}
Suppose the Markov chain $X$ satisfies \ref{cond:harris_ergodicity} - \ref{cond:geometric_ergodicity}. Let $g:\mathsf{X} \to \R^d$ be a function satisfying \ref{cond:integrability}. Let $\mathbf{w}_M \in \R^d$ be a random vector such that $\mathbf{w}_M \to \mathbf{w}=[w_1,\dots,w_d]$ $P_x$-almost surely for any $x \in \mathsf{X}$ and $w_i \ne 0$ for all $1\le i\le d$. Suppose $\boldsymbol{\delta}$ is given so that $0<\delta_{i}\le \Delta(F_{g_i})$. For any $1\le i,j\le d$, we have
\begin{enumerate}
\setlength\itemsep{0em}
    \item $\sum_{k\in\Z} \{\hat{\gamma}^{(\boldsymbol{\delta})}_{M,ij}(k;\mathbf{w}_M) - \gammaij(k)\}^2 \to 0 $
    \item  $ |\sum_{k\in\Z} \hat{\gamma}^{(\boldsymbol{\delta})}_{M,ij}(k;\mathbf{w}_M) -\sum_{k\in\Z} \gammaij(k)| \to 0$ 
\end{enumerate}
$P_x$-almost surely, as $M\to \infty$, for any initial condition $x \in \mathsf{X}$.
\end{thm}

The proof of Theorem~\ref{thm:est1-conv} is in Supplementary Material\ifshow~\ref{sec:proof_of_Theorem2}. \else~S2.2. \fi The second statement of Theorem~\ref{thm:est1-conv} implies as a Corollary the strong consistency of $\hat{\Sigma}_{pw}^{\boldsymbol{\delta}}$, based on the weights $\tilde{\mathbf{w}}_{M}$ in~\eqref{def:w_M}, for $\Sigma$:

\begin{cor}\label{cor:Sigma_pw_consistency}
Suppose the Markov chain $X$ satisfies \ref{cond:harris_ergodicity} - \ref{cond:geometric_ergodicity}. Let $g:\mathsf{X} \to \R^d$ be a function satisfying \ref{cond:integrability}. Assume that ${\rm Var}_\pi(g_j(X_0))>0$ for all $j=1,\dots,d$. Suppose $\boldsymbol{\delta}$ is given so that $0<\delta_{i}\le \Delta(F_{g_i})$ for $i=1,\dots,d$. Then the plug-in estimator $\hat{\Sigma}_{pw}^{(\boldsymbol{\delta})} \in \R^{d\times d}$ is $P_x$-strongly consistent for $\Sigma \in \R^{d\times d}$ for any $x\in \mathsf{X}$.
\end{cor}
In Corollary \ref{cor:Sigma_pw_consistency}, the condition ${\rm Var}_\pi(g_j(X_0))>0$ implies that the function $g_j$ is not $\pi$-a.e. constant, and is necessary in order to ensure $\tilde{w}_{Mi}$ in~\eqref{def:w_M} converge to the finite value $1/\text{Var}_{\pi}(g_i(X_0))$ for $i=1,...,d$.

We now move onto the convergence results for the refined positive semi-definite estimator $\hat{\Sigma}^{(\boldsymbol{\delta})}_{psd}$. We need the following Lemma, which ensures that the convergence of the momentLS estimator for $\sigma^2(\gamma_{\mathbf{c}^\top g})$ can be controlled uniformly over $\bc$. In particular, recall that $\hat{\Sigma}^{(\boldsymbol{\delta})}_{psd}$ estimator relies on estimating the asymptotic variance based on $\hat{\U}_{Mj}^\top g(X_t)$ for $j=1,\dots,d$. Lemma \ref{lem:sup_c_conv} allows us to control the distance $|\sigma^2(\Pi_\delta(\rinitf{\hat{\U}_{Mj}^\top g}))- \sigma^2(\gamma_{\hat{\U}_{Mj}^\top g})|$ based on the convergence in the worst-case direction.

\begin{lem}\label{lem:sup_c_conv}
    Suppose the Markov chain $X$ satisfies \ref{cond:harris_ergodicity} - \ref{cond:geometric_ergodicity}. Let $g:\mathsf{X} \to \R^d$ be a function which satisfies \ref{cond:integrability}.
    Suppose $\delta$ is chosen so that $0<\delta\leq\min_{1\le i \le d} \Delta(F_{g_i})$.
    We have $\sup_{\bc; \|\bc\|_2 =1}|\sigma^2(\Pi_\delta(\rinitf{\mathbf{c}^\top g}))- \sigma^2(\gamma_{\bc^\top g})|\to 0$ $P_x$-almost surely, for any $x\in \mathsf{X}$.
\end{lem}
The proof is deferred to Supplementary Material\ifshow~\ref{sec:proof_of_Lemma2}. \else~S2.3. \fi
The proof follows similar lines as in the proof of the strong convergence of $\sigma^2( \Pi_\delta(r_{\bc^\top g}))$ for fixed $\bc$, but requires some additional arguments to control the cross-products terms $\gammaij(\cdot)$ by diagonal terms $\gamma_{g_i}(\cdot)$ and $\gamma_{g_j}(\cdot)$ and obtain a bound which is independent of $\bc$ using Holder's inequality. Finally, we present the following Theorem which shows that the refined estimator $\hat{\Sigma}^{(\boldsymbol{\delta})}_{psd}$ is also strongly consistent for the true asymptotic variance $\Sigma$. The proof is deferred to Supplementary Material\ifshow~\ref{sec:proof_of_theorem3}. \else~S2.4. \fi
\begin{thm}\label{thm:mtv2}
    Assume the same conditions as in Corollary \ref{cor:Sigma_pw_consistency}. 
    The refined estimator $\hat{\Sigma}^{(\boldsymbol{\delta})}_{psd}$ converges $P_x$-almost surely to the true asymptotic variance $\Sigma$ for any initial condition $x \in \mathsf{X}$.
\end{thm} 
Finally, we summarize the strong consistency results for the momentLS auto(cross)-covariance sequence estimators $\hat{\gamma}^{(\boldsymbol{\delta})}_{M,ij}$ and asymptotic variance estimator $\hat{\Sigma}^{(\boldsymbol{\delta})}$ defined in \eqref{eq:gammaHatDef} and \eqref{eq:momentLSavar} respectively. This Corollary is a direct consequence of Theorem \ref{thm:est1-conv}, Corollary \ref{cor:Sigma_pw_consistency}, and Theorem \ref{thm:mtv2}.
\begin{cor}
     Assume the same conditions as in Corollary \ref{cor:Sigma_pw_consistency}. We have
     \begin{enumerate}
     \setlength\itemsep{0em}
         \item $\max_{1\le i,j \le d}\|\hat{\gamma}^{(\boldsymbol{\delta})}_{M,ij} - \gamma_{g_i,g_j}\|_2\to 0$, $P_x$-almost surely
         \item $\hat{\Sigma}^{(\boldsymbol{\delta})}$ is a positive semi-definite matrix, and $\|\hat{\Sigma}^{(\boldsymbol{\delta})} - \Sigma\|_F \to 0$, $P_x$-almost surely
     \end{enumerate}
     for each initial condition $x\in\mathsf{X}$.
\end{cor}
\section{Empirical studies}\label{sec:emp}
In this section, we conduct an empirical comparison between the proposed momentLS estimator $\hat{\Sigma}^{(\mathbf{\delta})}$ defined in~\eqref{eq:momentLSavar} and other state-of-the-art methods for estimating the asymptotic variance $\Sigma$ in the multivariate CLT. We assess the relative error  $\|\Sigma^{-1/2}(\hat{\Sigma} - \Sigma)\Sigma^{-1/2}\|_{F}$, which quantifies the relative distance between the true asymptotic variance matrix $\Sigma$ and the estimated variance matrix $\hat{\Sigma}$ produced by various methods. Additionally, we evaluate the empirical coverage probabilities of the confidence regions $C_\alpha(X; \hat{\Sigma})$ based on different estimates of $\hat{\Sigma}$. Here, for a chain $X=(X_0,\dots,X_{M-1})$, $C_\alpha(X; \hat{\Sigma})$ is defined as follows:
$
C_\alpha(X; \hat{\Sigma})=\left\{\mu \in \mathbb{R}^d: M\left(\hat{\mu}_{gM}-\mu\right)^\top \hat{\Sigma}^{-1}\left(\hat{\mu}_{gM}-\mu\right)<\chi^2_{1-\alpha, d}\right\},
$
where $\hat{\mu}_{gM} = M^{-1} \sum_{t=0}^{M-1} g(X_t)$ is the estimated mean computed from the chain $X$.

We evaluate the performance of the different methods through simulated and real-data examples.
For the simulated examples, we consider two setups: multivariate functions of a Markov chain generated using a finite state Metropolis-Hastings sampler, and a Vector AR1 process. In both cases, we can analytically compute the true asymptotic variances $\Sigma$. 
For real-data examples, we investigate a Bayesian logistic regression model using the BUPA's liver disorders dataset from the UCI machine learning repository.  We considered three popular samplers: the random walk Metropolis sampler, the No U-Turn Sampler \citep{hoffman2014no} from STAN \citep{stan2019} which is a dynamic variant of Hamiltonian Monte Carlo, and the Polya-Gamma Gibbs sampling scheme for Bayesian logistic regression introduced by~\citet{polson2013Bayesian}.

\subsection{Settings for simulated chains}
\paragraph{Metropolis-Hastings chain}
We consider a Metropolis-Hastings chain on a finite discrete state space $(\mathsf{X},2^{\mathsf{X}})$ where $\mathsf{X}=\{1,2,...,s\}$.
We randomly generated the target distribution $\pi: 2^{\mathsf{X}} \to [0,1]$ and the proposal distribution $P: \mathsf{X} \times 2^{\mathsf{X}} \to [0,1]$, where we drew a length $s$ random vector $U=[U_1,U_2,...,U_{s}]^\top$ and let $\pi(\{i\})=U_{i}/(\sum_{i'=1}^{s}U_{i'})$, and for each $i \in \mathsf{X}$ we drew a length $s$ random vector $V_i=[V_{i1},\dots,V_{is}]^\top$ and let $P(i,\{j\})=V_{ij}/\sum_{j=1}^{s}V_{ij}$. Then we constructed the Metropolis-Hastings transition kernel $Q$ based on $P$ and $\pi$. Finally, we simulated $g:\mathsf{X}\to \R^d$ by $g(i) \sim N(0,\Sigma_g)$ i.i.d. for $i \in \{1,\dots,s\}$ where $\textrm{diag}(\Sigma_g)  = [1,\dots,d^2]$ and $\Sigma_{g,ij} = \sqrt{\Sigma_{g,ii} \Sigma_{g,jj}}\rho^{|i-j|}$. With this choice of $g$, $g_j(X_t), j=1,\dots,d$ will be correlated with each other, and have different magnitudes. In the simulation study, $\pi$, $g$, and the transition kernel $Q$ were generated once, and then used for all of the discrete state Metropolis-Hastings chains, so that the same $\pi$, $g$, and $Q$ were used for each simulated chain.

With a slight abuse of notation, let $Q \in \R^{s\times s}$ denote the transition matrix defined as $Q_{ij} = Q(i, \{j\})$. Since $s=100$ is relatively small, we can compute the eigenvalues and eigenvectors for $Q$. Let $\lambda_1\geq \lambda_2\geq\cdots \geq\lambda_s$ denote the eigenvalues of $Q$ and $\phi_1,\dots \phi_s$ be the corresponding eigenvectors. Suppose the eigenvectors $\phi_i$ are normalized so that $\braket{\phi_i,\phi_j}_\pi = 1[i=j]$. Note we have $\lambda_1=1$ and $\phi_1 = \mathbf{1}_{s}$ since $Q\mathbf{1}_{s} = \mathbf{1}_{s}$. For each $i=1,\dots,s$, we can write $g_i$ and $\check{g}_i = g_i-E_\pi[g_i(X_0)]\mathbf{1}_{s}$ as
    $g_i(k) = \sum_{l=1}^s \braket{g_i,\phi_l}_\pi \phi_l(k)$ and
    $\check{g}_{i}(k) = \sum_{l=2}^s \braket{g_i,\phi_l}_\pi \phi_l(k)$
since $\braket{g_i,\phi_1}_\pi = E_\pi[g_i(X_0)]$. Then, 
     $\textrm{Cov}_\pi(g_i(X_0), g_j(X_k))=\braket{\check{g}_i, Q^{|k|}\check{g}_j}_\pi = \sum_{l=2}^{s}\braket{g_i,\phi_l}_\pi\braket{g_j,\phi_l}_\pi\lambda_l^{|k|}$
since $\langle \phi_l, \phi_m \rangle_\pi =1[l=m]$. 
In particular, we have
    $F_{g_i,g_j} = \sum_{l=2}^{s} \braket{g_i,\phi_l}_\pi \braket{g_j,\phi_l}_\pi \delta_{\{\lambda_l\}}$
 where $\delta_{\{a\}}$ denotes a unit point mass at $a$. Note $F_{g_i,g_j}$ is in general a signed measure on $(\R, \mathscr{B}_\R)$, since the weight $\braket{g_i,\phi_l}_\pi \braket{g_j,\phi_l}_\pi$ at each $\lambda_l$ is real, but not necessarily positive, except the case when $i=j$, where we have $\braket{g_i,\phi_l}_\pi \braket{g_i,\phi_l}_\pi = \langle g_i,\phi_l\rangle_\pi^2 \ge 0$.
Finally, we have the asymptotic variance $\Sigma$ where 
    $\Sigma_{ij} = \sum_{l=2}^{s} \braket{g_i,\phi_l}_\pi\braket{g_j,\phi_l}_\pi \frac{1+\lambda_l}{1-\lambda_l}.$

 \paragraph{Vector autoregressive process}
As a second example, we consider a vector autoregressive process of order $1$ (VAR(1)). Here we have a stationary Markov chain on a continuous state space $(\mathsf{X}, \mathscr{X}) = (\R^d,\mathscr{B}_{\R^d})$ where for $t=0,1,2, \ldots$, $X_t \in \R^{d}$, the process at time $t$ is defined by
$X_t= A X_{t-1}+\epsilon_t $
where $A$ is a $d \times d$ matrix, $\epsilon_t \stackrel{iid}{\sim} N(0, \Sigma_\epsilon)$, and $\Sigma_\epsilon$ is a $d \times d$ positive definite matrix. We let $A = D_\rho + (0.1)^2 (1_d 1_d^\top - I_d)$ where $D_\rho$ is a diagonal matrix, $1_d$ denotes the length $d$ column vector of 1's, and $\Sigma_\epsilon = I_d$. 
We consider two scenarios related to $D_\rho$: in the first case, we let $D_{\rho}={\rm diag}(0.9,0.9,-0.9,-0.9)$, resulting in each component of the chain being either strongly positively or negatively autocorrelated. In the second case, we let $D_{\rho}={\rm diag}(0.9,0.9,0.9,0.9)$, resulting in each component of the chain being all positively autocorrelated.

\citet{osawa1988reversibility} has showed that the VAR(1) process defined above is reversible if and only if $A\Sigma_\epsilon$ is symmetric. Moreover, if $\lambda_{\rm max}(A) <1$, the stationary distribution $\pi$ for the $X_t$ chain is $N(0_d, V)$, $0_d$ is the length $d$ column vector of $0$'s, $V = (I_d-A^2)^{-1} \Sigma_\epsilon$, and $I_d$ denotes the $n\times n$ identity matrix. The simulated process is reversible because $A$ is symmetric. Moreover, since $\operatorname{Cov}_\pi(X_k,X_0) =A\operatorname{Cov}_\pi(X_{k-1},X_0)=\dots= A^k V$, the asymptotic variance $\Sigma$ is 
    $\Sigma = \sum_{k\in\Z} A^k V = 2\sum_{k=0}^\infty A^kV - V = 2(I_d-A)^{-1}V- V.$

\subsubsection{Descriptions of estimators}

We compared the asymptotic variance estimation performance of the following estimators: 1. \textbf{(SV(Bartlett))} the spectral variance estimator with the modified Bartlett window $w_M(k)= (1-|k|/b_M) 1\{|k|<b_M\}$ \citep{Damerdji1991-oj,  flegal2010batch, vats2018strong}, 2. \textbf{(BM, OBM)} the batch means and overlapping batch means estimator \citep{vats2019multivariate}, 3. \textbf{(mtv-Init)} the multivariate initial sequence estimator from \cite{dai2017multivariate}, and 4. \textbf{(mtv-MLSE)} our multivariate moment least squares estimator $\hat{\Sigma}^{(\boldsymbol{\delta})}$  defined in \eqref{eq:momentLSavar}. 
To provide a clear exposition of the steps involved in computing $\hat{\Sigma}^{(\boldsymbol{\delta})}$, we have summarized the algorithm in Supplementary Material Section\ifshow~\ref{sec:alg_momentLS}. \else~S3. \fi All empirical experiments that we present in this section were performed using R software \citep{rlang}.  We used the \textbf{momentLS} R package available on GitHub 
\footnote{https://github.com/hsong1/momentLS} for computing univariate momentLS estimators. All other estimators were computed using the \textbf{mcmcse} R package \citep{mcmcse_R}.

Hyperparameters are required for the spectral variance estimator, BM, OBM, and momentLS estimators. A truncation point for SV(Bartlett) or a batch size for BM and OBM needs to be specified in advance. The vector $\boldsymbol{\delta}=[\delta_1,\dots,\delta_d]$ needs to be specified for momentLS estimators, which determines the moment spaces onto which each empirical autocovariance sequence $r_{gM,j}$, $j=1,\dots,d$ is projected. 
For SV(Bartlett), BM, and OBM, we used the truncation point (batch size) tuning method implemented in the R package {\bfseries mcmcse}~\citep{liu2021batch}. For the multivariate momentLS estimator, we utilized the tuning method proposed in \citet{berg2022efficient} to select $\delta_j$ for each univariate sequence $\{g_j(X_t)\}_{t\ge 0}$, defined as
    $\tilde{\delta}_{Mj} = 0.8 \frac{1}{L}\sum_{l=1}^L \hat{\delta}_{M/L,j}^{(l)}.$
Here, $\hat{\delta}_{M/L,j}^{(l)}$ is an estimator for $\Delta(F_{g_j})$ computed from the $l$th batch of empirical autocovariances. We used $L=5$. Subsequently we let $\boldsymbol{\delta} = [\tilde{\delta}_{M1},\dots, \tilde{\delta}_{Md}]$. The computational steps involved in computing $\tilde{\delta}_{Mj}$ are summarized in the \texttt{tune\_delta} subroutine in Supplementary Material Section\ifshow~\ref{sec:alg_momentLS}. \else~S3. \fi

\subsection{Results for simulated chains}
We compared the relative error and coverage rate from various methods (SV(Bartlett), BM, OBM, Mtv-Init, and Mtv-mLSE) with varying chain lengths $M\in \{5000, 10000, 20000, 40000, 80000\}$. In particular, at each sample size $M$, $B=1000$ independent chains were generated, and the various estimators were computed on each of these replicate chains. For each method and at each sample size $M$, we estimated the relative error and coverage rate as
   $\widehat{\textrm{relative error}}_{_M} = B^{-1}\sum_{b=0}^B\| \Sigma^{-1/2}(\hat{\Sigma}^{(b)}_{M} - \Sigma)\Sigma^{-1/2}\|_{F}$ and $
   \widehat{\textrm{coverage probability}}_{_M}  = B^{-1}\sum_{b=0}^B 1\{ \mu_{g} \in C_\alpha(X_M^{(b)}; \hat{\Sigma}^{(b)}_{M})\}$
where $X_M^{(b)}=(X_0^{(b)}, \dots,X_{M-1}^{(b)})$ is the $b$th simulated chain of length $M$ and $\hat{\Sigma}^{(b)}_{M}$ is the estimated asymptotic variance using $X_M^{(b)}$. We used $\alpha=0.05$ in $C_\alpha(X^{(b)}; \hat{\Sigma}^{(b)})$. 

 \begin{figure}[htb]
     \centering
    \includegraphics[trim=0 0 0 0]{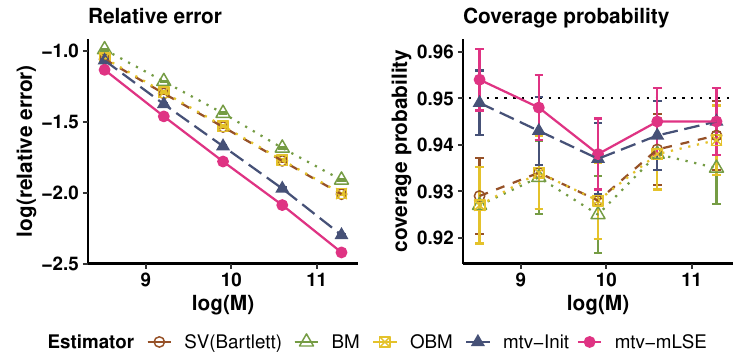}
     \caption{
     Estimated relative errors $\|\Sigma^{-1/2}(\hat{\Sigma} - \Sigma)\Sigma^{-1/2}\|_{F}$ and coverage probabilities $P_\pi(\mu_{gM} \in C_\alpha(X; \hat{\Sigma}))$ with $\alpha = 0.05$ for the Metropolis-Hastings example. Error bars indicate one standard error. The dotted horizontal line in the coverage probability plot indicates the $95$\% nominal coverage level. 
     }
     \label{fig:discrete}
 \end{figure}

 \begin{figure}[htb]
     \centering
     \begin{subfigure}[b]{\textwidth}
     \caption{Vector AR1 with $\bm{D_\rho = {\rm diag}(0.9,0.9,-0.9,-0.9)}$}
     \vspace{.5em}
     \centering
         \includegraphics[trim=0 0 0 0, clip]{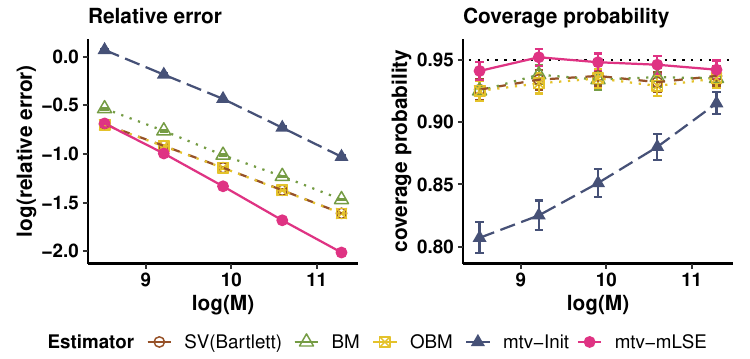}
     \end{subfigure}
    \begin{subfigure}[b]{\textwidth}
    \centering
    \caption{Vector AR1 with $\bm{D_\rho = {\rm diag}(0.9,0.9,0.9,0.9)}$}
    \vspace{.5em}
        \includegraphics[trim=0 0 0 0, clip]{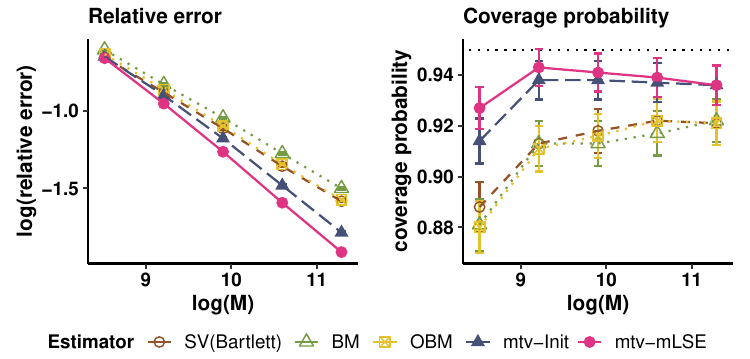}
    \end{subfigure}
    
     \caption{
     Estimated relative errors $\|\Sigma^{-1/2}(\hat{\Sigma} - \Sigma)\Sigma^{-1/2}\|_{F}$ and coverage probabilities $P_\pi(\mu_{gM} \in C_\alpha(X; \hat{\Sigma}))$ with $\alpha = 0.05$ for the two Vector AR1 examples. Error bars indicate one standard error. For coverage probability plots, the dotted horizontal lines indicate the $95$\% nominal coverage level. 
     }
     \label{fig:VAR1}
 \end{figure}
 
Figure \ref{fig:discrete} and \ref{fig:VAR1} display the estimated relative errors and coverage probabilities from the SV(Bartlett), BM, OBM, mtv-Init, and mtv-mLSE methods. The multivariate momentLS estimator (mtv-mLSE) appears to perform the best in terms of both relative error and coverage probability in all examples, demonstrating the smallest relative errors and empirical coverage rates closest to the 95\% nominal level across all sample sizes. Also, both mtv-mLSE and mtv-Init estimators seem to exhibit faster convergence rates when compared to the other methods, which was also the case for their counterparts in the univariate setting \citep{berg2022efficient}. 
While the multivariate initial sequence estimator (mtv-Init) performs second best in the Metropolis-Hastings example and the VAR1 example with all positively autocorrelated components, the estimator seems to struggle more in the VAR1 case with a mixture of positive and negatively autocorrelated components. In this case, the multivariate initial sequence estimator shows relatively large relative errors compared to the other methods and tending to provide under-coverage of the the true parameter. Both the spectral variance (SV(Bartlett)) and overlapping batch means (OBM) estimators outperform the batch means (BM) method, with SV and OBM performing very similarly in the Metropolis-Hastings and both VAR1 examples.

\subsection{Bayesian logistic regression } \label{sec:BayesianLogisticExample}
In this section, we fit a Bayesian logistic regression model to the Liver Disorders dataset sourced from the UCI machine learning repository.  The analysis focuses on assessing the relationship between blood test panel predictors (MCV, ALKPHOS, SGPT, SGOT, GAMMAGT) and an individual's drinking habits. The dataset consists of $n=341$ observations after removing duplicates. Each row represents an individual's blood test results and their average daily alcohol intake. The goal is to predict an individual's drinking status, which is classified as ``heavy" if their average daily drink count exceeds 3, based on the $5$ blood test panel results represented as $\mathbf{x}_i = (x_{i1},\dots,x_{i5})\in \R^5$. We denote the binary response variable as $Y_i$, where $Y_i=1$ if the average daily drink count of the $i$th individual exceeds 3 and 0 otherwise. We suppose 
    $Pr(Y_i=1) =  \frac{\exp(\alpha + \sum_{j=1}^5 \beta_j x_{ij}) }{1+\exp(\alpha + \sum_{j=1}^5 \beta_j x_{ij})}$
and assign a $N(0,5)$ prior on the intercept $\alpha$ and independent $N(0,1)$ priors on the coefficients $\boldsymbol{\beta} = (\beta_1,\dots,\beta_5)$. 

We sample $\{\boldsymbol{\beta}(t)\}_{t=0}^{M-1}$ from the posterior distribution $\boldsymbol{\beta}| \{Y_i\}_{i=1}^{n} \sim \pi(\cdot)$ using three popular samplers: 1. a random walk Metropolis sampler, 2. the No U-Turn sampler \citep{hoffman2014no}, and 3. the Polya-Gamma Gibbs sampler \citep{polson2013Bayesian}. We considered varying chain lengths $M \in \{2500, 5000, 10000, 20000, 40000\}$ for each sampler. Similar to previous examples, we generated $B=1000$ parallel chains from each sampler and estimated relative error and coverage probability. 

For the random walk Metropolis sampler (RW-metrop), we used the MCMClogit function in the \textbf{MCMCpack} R package \citep{JSSv042i09} with a burn-in period of $5000$ iterations. Regarding the proposal step, the proposal distribution at time $t$ is taken to be $N(\boldsymbol{\beta}(t), V_{prop})$ where $V_{prop}=T\left(B_0+C^{-1}\right)^{-1} T$, where $T = 1.1 I_d$ is a diagonal positive definite matrix, $B_0 = \textrm{diag}(0.2,1,1,1,1,1)$ is the prior precision matrix, and $C$ is the asymptotic variance-covariance matrix of the MLEs.  

The No U-Turn sampler (NUTS) is a dynamic variant of the Hamiltonian Monte Carlo (HMC) method, which itself is a Metropolis-Hastings algorithm where proposal states are generated by integrating a system of Hamiltonian equations using the leapfrog integrator. 
HMC methods require the specification of three algorithm parameters: 1. the discretization time (step size) $\epsilon$, 2. the mass matrix $M_{mass}$, and 3. the number of leapfrog steps $L$. \citet{hoffman2014no} introduced the No U-Turn sampler (NUTS), which automatically selects the value of $L$ during the sampling process.
The No U-Turn sampler implemented in STAN employs a warmup and sampling phases. During the warmup phase, the sampler optimizes $\epsilon$ and $M_{mass}$ to match an acceptance rate target. In the subsequent sampling phase, the sampler generates Markov chain samples based on the prescribed transition mechanism, which depends on the HMC parameters $(\epsilon, M_{mass})$. Since we want to generate parallel chains from the same transition mechanism, we first learned $\epsilon$ and $M_{mass}$ using $100000$ warm-up samples, and subsequently generated $B$ parallel chains of length $M+5000$ (with $0$ warmup iteration) using the same $(\epsilon, M_{mass})$. The first $5000$ iterations were treated as iterations from a burn-in period, and excluded in the final chains. All sampling was performed using the \textbf{CmdStan}, the command-line interface to the Stan statistical modeling language \citep{stan2019}, as the current R implementation of Stan does not allow direct specification of $\epsilon$ and $M_{mass}$ during the sampling phase. 

The Polya-Gamma Gibbs sampler (PG) alternates draws of independent Polya-Gamma random variables and draws from a multivariate normal distribution to update $\beta$. For drawing the required Polya-Gamma random variables, we used the Polya-Gamma sampler implemented in \textbf{BayesLogit}~\citep{polson2013Bayesian}. The Polya-Gamma Gibbs sampler has been shown to be uniformly ergodic~\citep{choi2013polya}, and therefore this chain satisfies our geometric ergodicity requirement~\ref{cond:geometric_ergodicity}.

\paragraph{Ground truth}
To compare estimated asymptotic variances and coverage probabilities from the competing methods for the Bayesian logistic regression example, we also need accurate reference estimates of the posterior mean and asymptotic variance for each coefficient.  
Since both quantities are unknown, for each sampler sp $\in \{\textrm{RW-metrop,NUTS,PG}\}$, we independently generated $B_2=10000$ independent chains $\{\boldsymbol{\beta}_{\rm idp,sp}^{(b)}(t)\}_{t=0}^{M_2-1}$ of length $M_2=40000$ (after burn-in periods of length $5000$) to estimate true mean and asymptotic variance. We let $\bar{\boldsymbol{\beta}}_{\textrm{idp,sp}}^{(b)}$ refer to the sample mean value of $\boldsymbol{\beta}$ from the $b$th chain for $b=1,\dots,B_2$, and we use $\boldsymbol{\beta}_{\textrm{orcl,sp}} = \frac{1}{B_2}\sum_{b=1}^{B_2}\bar{\boldsymbol{\beta}}_{\textrm{idp,sp}}^{(b)}$ to estimate the posterior mean, and $\Sigma_{\textrm{orcl,sp}}=\frac{M_2}{B_2-1}\sum_{b=1}^{B_2}(\bar{\boldsymbol{\beta}}_{\textrm{idp,sp}}^{(b)}-\boldsymbol{\beta}_{\rm orcl, sp})(\bar{\boldsymbol{\beta}}_{\textrm{idp,sp}}^{(b)}-\boldsymbol{\beta}_{\rm orcl, sp})^\top$ to estimate the asymptotic variance matrix from the sampler sp $\in \{\textrm{RW-metrop,NUTS,PG}\}$. 

We note that since each sampler targets the same posterior, the true mean value of $\beta$ with respect to the stationary distribution is the same for each sampler. In contrast, the asymptotic variance matrices for the MCMC sampling will generally be different for the different samplers, as the different sampling mechanisms typically result in different autocovariances at non-zero lags in~\eqref{eq:avar_sum}. We present the estimated ground truth mean vectors $\boldsymbol{\beta}_{\textrm{orcl,sp}}$ and asymptotic variance matrices $\Sigma_{\textrm{orcl,sp}}$ for the 3 samplers sp $\in \{\textrm{RW-metrop,NUTS,PG}\}$ in Tables\ifshow~\ref{tab:trueMean} \else~S1 \fi and\ifshow~\ref{tab:avarComp}\else~S2\fi, respectively, of the Supplementary Material Section\ifshow~\ref{sec:oracle_values}. \else~S4. \fi

Figure \ref{fig:bayesLogistic} shows the estimated relative errors and coverage probabilities for the three samplers. Similarly to the results in the simulated examples, the multivariate momentLS estimator (mtv-mLSE) continues to perform the best in terms of both relative error and coverage probability for all three samplers on average. It consistently achieves the smallest relative errors and empirical coverage rates closest to the 95\% nominal level across all sample sizes except when $M=2500$. 
As in the previous case, the SV(Bartlett) and OBM estimators performed quite similarly in both examples, consistently outperforming the batch means (BM) method.
While the multivariate initial sequence estimator (mtv-Init) shows the second best performance for all three samplers when the considered chain length is the largest ($M=40000$), it appears to show relatively worse relative error performance compared to SV(Bartlett)and OBM estimator for the No U-Turn and Polya-Gamma samplers across the considered chain lengths ranging from $2500$ to $40000$.

\paragraph{Computation time}
The computational times required to estimate the asymptotic variances from each sampler vary depending on the variance estimator used. Table~\ref{tab:computationTime} shows the average computation time in seconds for the Bayesian example, with chain lengths of $M=10000$ and $M=40000$, each averaged over $100$ executions on a personal laptop with an Apple M1 chip and 16GB RAM. 

The momentLS estimator, which involves point-wise element estimation, was the slowest, taking approximately 0.2 seconds to compute the asymptotic variance for a chain length of $40000$. Despite its slower computation time compared to other estimators, we believe the momentLS estimator's computation times are still reasonable and that the estimator will be practically applicable in realistic problems.

\begin{table}[ht]
\centering
\small
\begin{tabular}{lc|ccccc}
  \hline
Sampler & M & SV(Bartlett) & BM & OBM & mtv-Init & mtv-mLSE \\ 
  \hline
RW-metrop & 10000 & 0.007 (0.000) & 0.002 (0.000) & 0.009 (0.000) & 0.024 (0.000) & 0.118 (0.003) \\ 
  NUTS & 10000 & 0.012 (0.000) & 0.007 (0.000) & 0.008 (0.000) & 0.003 (0.000) & 0.063 (0.003) \\ 
  PG & 10000 & 0.009 (0.002) & 0.002 (0.000) & 0.004 (0.000) & 0.003 (0.000) & 0.052 (0.000) \\ \hline
  RW-metrop & 40000 & 0.037 (0.004) & 0.007 (0.000) & 0.050 (0.000) & 0.094 (0.000) & 0.246 (0.004) \\ 
  NUTS & 40000 & 0.054 (0.003) & 0.027 (0.000) & 0.036 (0.002) & 0.013 (0.000) & 0.227 (0.006) \\ 
  PG & 40000 & 0.027 (0.000) & 0.008 (0.002) & 0.014 (0.000) & 0.012 (0.000) & 0.202 (0.006) \\ 
   \hline
\end{tabular}
\caption{Average computation times and standard errors (in seconds) for the various estimators for the Bayesian logistic regression example. The times shown are the average timings from the variance estimation portion of the computation for $100$ executions, at chain length $M \in \{10000,40000\}$.\label{tab:computationTime}} 
\end{table}

In addition, to demonstrate how the computation time of our estimator scales with the dimensions of the chain, we further simulated chains from the Bayesian example above with $5$ predictors. Instead of using raw predictors, we used natural cubic splines with degrees of freedom = $k$ for each predictor for $k=1,\dots,10$. For $k>1$, the knots are placed at the $1/k,\dots,(k-1)/k$ quantiles of each predictor (for $k=1$, there are no knots). This results in a model with $d=5k+1$ parameters, leading to $d$-dimensional chains for $d$ ranging from $6$ to $51$. We simulated length $M=10000$ chains from the NUTS sampler and measured the average computation times for all estimators over $100$ executions.

\begin{table}[ht]
\centering
\small
\begin{tabular}{lcccc}
  \hline
 & $d=6$ & $d=16$ & $d=36$ & $d=51$ \\ 
  \hline
SV & 0.009 (0.002) & 0.021 (0.002) & 0.057 (0.003) & 0.089 (0.005) \\ 
  BM & 0.002 (0.000) & 0.009 (0.002) & 0.029 (0.003) & 0.042 (0.004) \\ 
  OBM & 0.003 (0.000) & 0.013 (0.002) & 0.046 (0.002) & 0.084 (0.004) \\ 
  mtv-Init & 0.004 (0.000) & 0.018 (0.000) & 0.090 (0.001) & 0.171 (0.000) \\ 
  mtv-mLSE & 0.056 (0.001) & 0.322 (0.006) & 1.682 (0.008) & 3.766 (0.014) \\ 
   \hline
\end{tabular}
\caption{Average computation times and standard errors (in seconds) for Bayesian logistic regression example with natural cubic splines predictors with varying degrees of freedom $k$ (the chain dimension $d=5k+1)$. The times shown are the average timings from the variance estimation portion of the computation for $100$ executions at chain length $M=10000$.\label{tab:computationTime2}}
\end{table}

\section{Conclusion}\label{sec:concl}
In this work, we have introduced novel multivariate momentLS sequence and asymptotic variance estimators based on univariate momentLS sequence and asymptotic variance estimators proposed in \citet{berg2022efficient}. Our approach is based on the observation that the cross-covariance sequence admits a similar mixture representation to the autocovariance sequence, with a signed measure as the mixing measure, and can be expressed as the difference between two moment sequences. We have shown that the supports of the representing measures of the cross-covariance sequences can be controlled by the supports of the representing measures of each autocovariance sequence. Therefore we can use the hyperparameter $\boldsymbol{\delta}$ based on estimating $\Delta(F_{g_j})$ from each univariate chain $g_j(X_t)$, $j=1,\dots,d$. 
Theoretically, we have established strong consistency results for both the sequence and asymptotic variance MomentLS estimators, provided that each $\boldsymbol{\delta}_i$ is chosen to be sufficiently small. Furthermore, through simulated and real-data Bayesian logistic regression examples, we have demonstrated that our estimator empirically outperforms the competing methods.

\begin{figure}[H]
     \centering
     \begin{subfigure}[b]{\textwidth}
     \centering
     \caption{Random Walk Metropolis sampler}
     \vspace{.5em}
         \includegraphics[trim=0 0 0 0, clip]{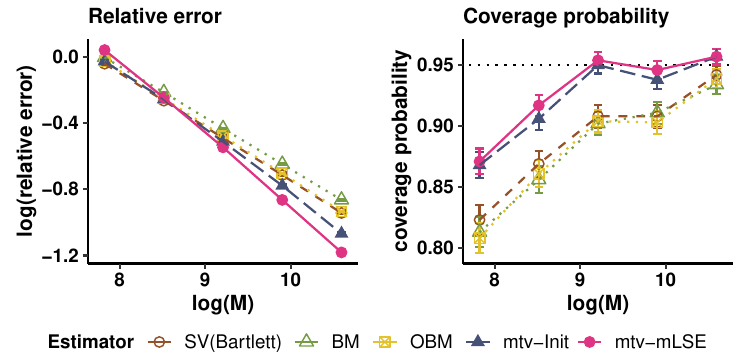}
     \end{subfigure}
\begin{subfigure}[b]{\textwidth}
\centering
\caption{No U-Turn sampler}
\vspace{.5em}
    \includegraphics[trim=0 0 0 0, clip]{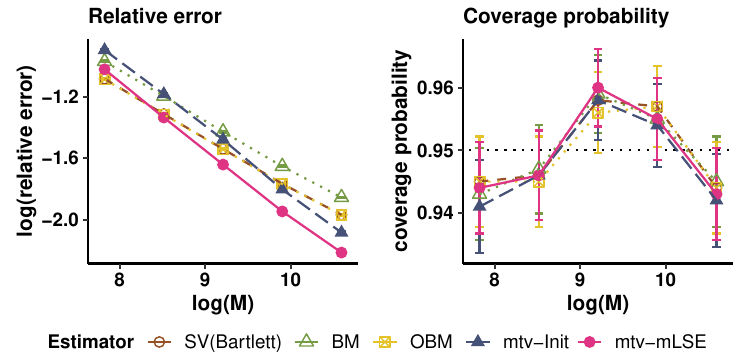}
\end{subfigure}
\begin{subfigure}[b]{\textwidth}
\centering
\caption{Polya-Gamma Gibbs sampler}
\vspace{.5em}
    \includegraphics[trim=0 0 0 0, clip]{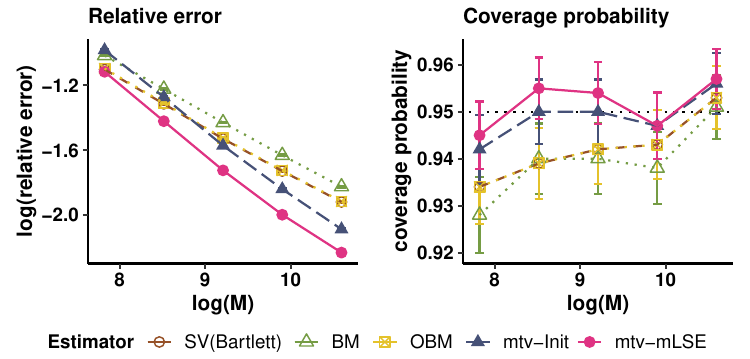}
\end{subfigure}
     \caption{
     Estimated relative errors $\|\Sigma^{-1/2}(\hat{\Sigma} - \Sigma)\Sigma^{-1/2}\|_{F}$ and coverage probabilities $P_\pi(\mu_{gM} \in C_\alpha(X; \hat{\Sigma}))$ with $\alpha = 0.05$  for the Bayesian logistic model example using the    
     random walk Metropolis sampler, the No U-Turn sampler, and the Polya-Gamma Gibbs sampler. Error bars indicate one standard error. For coverage probability plots, the dotted horizontal lines indicate the $95$\% nominal coverage level. 
     }
     \label{fig:bayesLogistic}
 \end{figure}

\section{Acknowledgements}
HS and SB gratefully acknowledge support from NSF DMS-2311141. The authors would like to thank the anonymous reviewers for their helpful comments and feedback. 

\section{Disclosure Statement}
The authors report there are no competing interests to declare.



\ifshow\else
\bibliographystyle{plainnat}
\bibliography{bib}
\fi
\iftwofiles\putbib
\end{bibunit}\fi

\ifshow
\newpage
\iftwofiles\begin{bibunit}\fi
\setcounter{page}{1}
\renewcommand{\thepage}{S\arabic{page}} 
\renewcommand{\thesection}{S\arabic{section}}  
\renewcommand{\thetable}{S\arabic{table}}  
\renewcommand{\thefigure}{S\arabic{figure}}
\renewcommand{\theequation}{S-\arabic{equation}}
\setcounter{equation}{0}
\setcounter{section}{0}
\setcounter{table}{0}
\spacingset{1.5}
\begin{center}
{\Large\bf Supplement to ``Multivariate moment least-squares variance estimators for reversible Markov chains"}\\
\vspace{1em}
{\large Hyebin Song and Stephen Berg}\\
{\large Department of Statistics, Pennsylvania State University}
\end{center}

\section{Proofs for Section \ref{sec:mtv_momentLS}}
\subsection{Proof for Proposition \ref{prop:rep_gamma_ij}}\label{sec:proof_of_prop1}
\begin{proof}
Let $\mathcal{H} = L^2(\mathsf{X},\mathscr{X},\pi) = L^2(\pi)$ be the space of square integrable functions with respect to $\pi$, with inner product $\langle \cdot, \cdot \rangle_\pi$. Let $\mathscr{B}(\mathcal{H})$ be the Banach algebra of all bounded linear operators $T$ on $\mathcal{H}$, with norm $\vertiii{T}_{L^2(\pi)} = \sup \{\|Tf\|_{L^2(\pi)}: f \in \mathcal{H}, \|f\|_{L^2(\pi)} \le 1\}$. We recall that $\gamma_{g_i,g_j}(k) = {\rm Cov}_\pi (g_i (X_0), g_j (X_k)) = \langle g_i, Q_0^k g_j \rangle_\pi =  \langle Q_0^{|k|} g_i,  g_j \rangle_\pi$ for $k\ge 0$ since $Q_0$ is self-adjoint. When $k<0$, $\gamma_{g_i,g_j}(k) = {\rm Cov}_\pi (g_i (X_{|k|}), g_j (X_0)) = \langle Q_0^{|k|} g_i,  g_j \rangle_\pi$.

Since $Q_0$ is self-adjoint and $\vertiii{Q_0}_{L^2(\pi)}\le 1$, there exists a spectral resolution $\{E(\lambda)\}$ of $Q_0$ (e.g., Theorem 6.1 in \citealt{stein2009real}), which is a function from $\R$ to $\mathscr{B}(\mathcal{H})$, such that 
\begin{align}\label{eq:spectral_decomp}
Q_0=\int_{[a,b]} \lambda dE(\lambda),
\end{align}
for $a = \min_{f\in \mathcal{H}; \|f\|_{L^2(\pi)}\le 1} \langle Q_0f,f \rangle_\pi$ and $b = \max_{f\in \mathcal{H}; \|f\|_{L^2(\pi)}\le 1} \langle Q_0f,f \rangle_\pi$ (note that both $a$ and $b$ are real values due to self-adjointness of $Q_0$), which means that
for any $f,g \in \mathcal{H}$,
\begin{align*}
    \langle Q_0 f,g \rangle_\pi = \int \lambda \mE{f}{g} (d\lambda)
\end{align*}
where $\mE{f}{g}$ is the complex-valued measure on $\mathscr{B}_\R$ uniquely defined by the function $\Psi_{f,g}:\lambda \to \langle E(\lambda)f,g\rangle_\pi$. 
We note that $\mE{f}{g}$ is well-defined since $\Psi_{f,g}$ is of bounded variation: by the polarization identity, for $f,g\in \mathcal{H}$, 
\begin{align*}
    \Psi_{f,g}(\lambda) = \langle E(\lambda)f,g\rangle_\pi = \langle E(\lambda)f,E(\lambda) g\rangle_\pi =Re(\langle E(\lambda)f,E(\lambda)g\rangle_\pi) + i Im( \langle E(\lambda)f,E(\lambda)g\rangle_\pi)
\end{align*}
where
\begin{align}
    &Re(\langle E(\lambda)f,E(\lambda)g\rangle_\pi) = \frac{1}{4} \{\|E(\lambda)(f+g)\|_{L^2(\pi)} - \|E(\lambda)(f-g)\|_{L^2(\pi)}\}\nonumber\\
    &Im(\langle E(\lambda)f,E(\lambda)g\rangle_\pi) = \frac{1}{4} \{\|E(\lambda)(f+ig)\|_{L^2(\pi)} - \|E(\lambda)(f-ig)\|_{L^2(\pi)}\},\label{eq:complex_polarization_identity}
\end{align}
and $\lambda \to \|E(\lambda)f\|_{L^2(\pi)}$ is an increasing function of $\lambda$ for any $f \in\mathcal{H}$ by the properties of spectral resolutions.
Moreover, the complex measure $\mE{f}{g}$ is finite, i.e., $|\mE{f}{g}|(\R) <\infty$ where we use $|\mu|$ denote the total variation of a complex measure $\mu$. Note that $|\mE{f}{g}|$ is the same as the positive measure defined by the total variation function $T_{\Psi_{f,g}}$ of $\Psi_{f,g}$ (e.g., Theorem 3.29 in \citealt{folland1999real}). By the polarization inequality \eqref{eq:complex_polarization_identity}, for any $x \in \R$,
\begin{align*}
|\mE{f}{g}|((-\infty,x]) &\le \frac{1}{4}\{\Psi_{f+g,f+g}(x) + \Psi_{f-g,f-g}(x)  + \Psi_{f+ig, f+ig} (x) + \Psi_{f-ig,f-ig}(x)\}\\
    &\le \frac{1}{4}\sum_{j=1}^4 \|f+a_j g\|_{L^2(\pi)}^2\,\,\textrm{where}\, a_j = \pm 1 \,\textrm{or}\, \pm i  \\
    &\le \{\|f\|_{L^2(\pi)}^2 + \|g\|_{L^2(\pi)}^2 + 2\|f\|_{L^2(\pi)}\|g\|_{L^2(\pi)}\} < \infty.
\end{align*}
In particular, $|\mE{f}{g}|(\R) \le \|f\|_{L^2(\pi)}^2 + \|g\|_{L^2(\pi)}^2 + 2\|f\|_{L^2(\pi)}\|g\|_{L^2(\pi)} < \infty$.

Let $f=g_i$ and $g=g_j$, which are real-valued functions by assumption. For a function $\Phi$ continuous on $[a,b]$, the operator $\Phi(Q_0)$ defined as the unique limit of approximating polynomials has the following spectral representation
\begin{align}\label{eq:functional_calculus}
    \langle \Phi(Q_0) f,g \rangle_\pi = \int \Phi(\lambda) \mE{f}{g} (d\lambda)
\end{align}
for any $f,g \in \mathcal{H}$ \citep{stein2009real}. Let $\Phi(\lambda) = \lambda^{|k|}$. Then 
\begin{align*}
    \gamma_{g_i,g_j}(k) =\langle Q_0^{|k|} g_i,  g_j \rangle_\pi = \int \lambda^{|k|} F_{g_i,g_j}(d\lambda)
\end{align*}
where we recall that $\mE{g_i}{g_j}$ is the complex-valued measure defined by $\Psi_{g_i,g_j}:\lambda \to \langle E(\lambda)g_i,g_j\rangle_\pi$.

We show that $\mE{g_i}{g_j}$ is a real-valued measure for real-valued $g_i,g_j\in L_2(\pi)$, and non-negative if $g_i = g_j$. For a real-valued function $f$, $Q_0f(x) = \int Q_0(x,dy) f(y)$ is also real-valued. This implies for the operator $\varphi(Q_0)$ which is defined based on a real-valued function $\varphi:[a,b] \to \R$, $\varphi(Q_0)f$ is also a real-valued function. In particular, for any $\lambda \in \R$, 
\begin{align*}
E(\lambda)g_i = \begin{cases}
    0 & \lambda < a\\
    \varphi^\lambda(Q_0)g_j  & a \le \lambda \le b\\
    g_j &\lambda>b
\end{cases}
\end{align*}
is real-valued where $\varphi^\lambda(t) = 1\{t\le \lambda\}$ (c.f. p310 in \citep{stein2009real}). Therefore, $\Psi_{g_i,g_j}(\lambda)=F_{g_i,g_j}((-\infty, \lambda]) = Re(F_{g_i,g_j}((-\infty, \lambda])) $, i.e., $F_{g_i,g_j}$ is a signed measure. When $g_i=g_j$, $\Psi_{g_i,g_i}(\lambda) = \langle E(\lambda)g_i, g_i\rangle_\pi =\langle E(\lambda)g_i, E(\lambda)g_i\rangle_\pi  = \| E(\lambda)g_i \|_{L^2(\pi)}^2  \ge 0$, and $F_{g_i,g_i}$ is a positive measure.

We show that the support of $\mE{g_i}{g_j}$ is contained in $[-1+\delta(Q), 1-\delta(Q)]$. Here, the support of $\mE{g_i}{g_j}$ is defined as ${\rm Supp}(\mE{g_i}{g_j}) = {\rm Supp}(\mE{g_i}{g_j}^{+}) \cup {\rm Supp}(\mE{g_i}{g_j}^{-})$ where $F_{g_i,g_j}^+$ and $F_{g_i,g_j}^{-}$ denote the positive and negative parts of $F_{g_i,g_j}$ in the Jordan-Hahn decomposition of $F_{g_i,g_j}$. First, we note that ${\rm Supp}(\mE{g_i}{g_j}) \subseteq [a,b]$. This is because the distribution functions of $\mE{g_i}{g_j}^{+}$ and $\mE{g_i}{g_j}^{-}$ are $\frac{1}{2}(T_{\Psi_{g_i,g_j}} + \Psi_{g_i,g_j})$ and $\frac{1}{2}(T_{\Psi_{g_i,g_j}} - \Psi_{g_i,g_j})$ respectively, where $T_{\Psi_{g_i,g_j}}$ is the total variation function of $\Psi_{g_i,g_j}$, and both $\Psi_{g_i,g_j}$ and $T_{\Psi_{g_i,g_j}}$ are constant on $(-\infty,a)$ and $(b,\infty)$. 
Now we show $[a,b] \subseteq [-1+\delta(Q), 1-\delta(Q)]$. Since $Q_0$ is self-adjoint, $\max\{|\lambda|; \lambda\in \sigma(Q_0)\} = \vertiii{Q_0}_{L^2(\pi)}$, and therefore the spectral gap $\delta(Q) = 1-\vertiii{Q_0}_{L^2(\pi)} = 1-\max\{|a|,|b|\}$ by Proposition 6.2 in \citet{stein2009real}. Then we have
$[a,b] \subseteq [-\max\{|a|,|b|\},\max\{|a|,|b|\}]=[-1+\delta(Q), 1-\delta(Q)]$.

Finally, we discuss the uniqueness of $\mE{g_i}{g_j}$. Suppose there exists another finite signed measure $\tilde{F}$ with $\operatorname{Supp}(\tilde{F})\subseteq [-1+\delta(Q), 1-\delta(Q)]$ such that for all $k \in \N$
\begin{align*}
    \int \lambda^k  F_{g_i,g_j}(d\lambda)  =\int  \lambda^k  \tilde{F} (d\lambda).
\end{align*}
This implies for any real-valued polynomial $p_N$, 
\begin{align}\label{eq:prop1:poly}
    \int p_N(\lambda)  F_{g_i,g_j}(d\lambda)  =\int  p_N(\lambda)  \tilde{F} (d\lambda).
\end{align}
We show that $\mE{g_i}{g_j}((-\infty,t]) = \tilde{F}((-\infty, t])$ for any $t \in \R$. This implies that $\mE{g_i}{g_j} = \tilde{F}$ (e.g., Theorem 3.29 in \cite{folland1999real}). 

Let $a_0 = -1+\delta(Q)$ and $b_0 = 1-\delta(Q)$ for notational simplicity. 
We have that $\mE{g_i}{g_j}((-\infty,t]) = \tilde{F}((-\infty, t]) =0$ for $t<a_0$ and $\mE{g_i}{g_j}((b_0,\infty)) = \tilde{F}((b_0,\infty)) =0$. Taking $p_N(\lambda)=1$ in \eqref{eq:prop1:poly}, we have $\mE{g_i}{g_j} ((-\infty, b_0]) = \tilde{F}((-\infty, b_0])$, and therefore, $\mE{g_i}{g_j}((-\infty,t]) = \tilde{F}((-\infty, t])$ for all $t \ge b_0$. Let $t \in [a_0,b_0)$ and $\epsilon>0$ be given.  
Choose $\epsilon_\lambda>0 $ such that both $|\mE{g_i}{g_j}|((t,t+\epsilon_\lambda]) \le \epsilon/4$  and $|\tilde{F}|((t,t+\epsilon_\lambda])\le \epsilon/4$  (such $\epsilon_\lambda>0$ exists from the continuity properties of measures). Define a continuous and bounded function $h_t$ such that
\begin{align*}
    h_t(\lambda) = \begin{cases}
        0 & \lambda \le a_0-\epsilon_\lambda\\
        \frac{1}{\epsilon_\lambda} (\lambda-(a_0-\epsilon_\lambda)) & a_0-\epsilon_\lambda <\lambda \le a_0\\
        1 & a_0 < \lambda \le t\\
        -\frac{1}{\epsilon_\lambda} (\lambda - (t+\epsilon_\lambda)) & t<\lambda \le t+\epsilon_\lambda\\
        0 & \lambda > t+\epsilon_\lambda.
    \end{cases}
\end{align*}
Note that $h_t$ and $1_{(-\infty,t])}$ only differ on $(-\infty, a_0)$ and $(t,t+\epsilon_\lambda)$ by definition of $h_t$, and therefore
\begin{align}
    \int |h_t(\lambda) - 1_{(-\infty,t]}(\lambda)| |\mE{g_i}{g_j}|(d\lambda) =\int_{(t,t+\epsilon_\lambda)}|h_t(\lambda) - 1_{(-\infty,t]}(\lambda)| |\mE{g_i}{g_j}|(d\lambda) \le |\mE{g_i}{g_j}|((t, t+\epsilon_\lambda]) \le \epsilon/4.\label{eq:prop1:ineq1}
\end{align}
Similarly, 
\begin{align}\label{eq:prop1:ineq2}
    \int |h_t(\lambda) - 1_{(-\infty,t])}| |\tilde{F}|(d\lambda)\le \epsilon/4
\end{align}
Let $C_0 = \max\{|\mE{g_i}{g_j}|(\R), |\tilde{F}|(\R)\}$. Note $C_0<\infty$ because $|\mE{g_i}{g_j}|(\R) \le \|g_i\|_{L^2(\pi)}^2 + \|g_j\|_{L^2(\pi)}^2 + 2\|g_i\|_{L^2(\pi)}\|g_j\|_{L^2(\pi)} < \infty$ and $\tilde{F}$ is assumed to be a finite measure.
By the Weierstrass approximation theorem, we can find a polynomial $p_N(\lambda) = \sum_{l=0}^k \alpha_l \lambda^k$ with $\alpha_1,\dots,\alpha_N \in \R$ on $[a,b]$ such that $|h_t(\lambda) - p_N(\lambda)|\le \epsilon/(2C_0)$ for all $\lambda\in[a,b]$.
Let $\Delta_1(\lambda) = |h_t(\lambda) - 1_{(-\infty,t]}(\lambda)|$ and $\Delta_2(\lambda) = |h_t(\lambda) - p_N(\lambda)|$.
We have,
\begin{align*}
    &|\mE{g_i}{g_j}((-\infty,t]) - \tilde{F}((-\infty,t])| \\
    &=|\int 1_{(-\infty,t]}(\lambda) \mE{g_i}{g_j}(d\lambda)-\int 1_{(-\infty,t]}(\lambda) \tilde{F}(d\lambda) | \\
    &\le |\int \{1_{(-\infty,t]}(\lambda)-h_t(\lambda) \} \mE{g_i}{g_j}(d\lambda) -\int \{1_{(-\infty,t]}(\lambda)-h_t(\lambda)\} \tilde{F}(d\lambda) | \\
    &\qquad + |\int \{h_t(\lambda)-p_N(\lambda) \} \mE{g_i}{g_j}(d\lambda) -\int \{h_t(\lambda)-p_N(\lambda)\} \tilde{F}(d\lambda) | \\
    & \qquad +  |\int p_N(\lambda) \mE{g_i}{g_j}(d\lambda)-\int p_N(\lambda) \tilde{F}(d\lambda)|\\
    &\le \int \Delta_1(\lambda) |\mE{g_i}{g_j}|(d\lambda) +\int \Delta_1(\lambda) |\tilde{F}|(d\lambda) +\int \Delta_2(\lambda) |\mE{g_i}{g_j}|(d\lambda) +\int \Delta_2(\lambda) |\tilde{F}|(d\lambda) 
    \end{align*}
We have $\int \Delta_1(\lambda) |\mE{g_i}{g_j}|(d\lambda) +\int \Delta_1(\lambda) |\tilde{F}|(d\lambda) \le \epsilon/2$ by \eqref{eq:prop1:ineq1} and \eqref{eq:prop1:ineq2}.    Moreover, we have,  \begin{align*}
         \int \Delta_2(\lambda) |\mE{g_i}{g_j}|(d\lambda) +\int \Delta_2(\lambda) |\tilde{F}|(d\lambda)  \le \epsilon/(2C_0) \{ |\mE{g_i}{g_j}|([a,b])+|\tilde{F}|([a,b])\} = \epsilon/2
    \end{align*}
since $\sup_{\lambda\in [a,b]} \Delta_2(\lambda) \le \epsilon/(2C_0)$. Therefore, since $\epsilon$ is arbitrary, we have $\mE{g_i}{g_j}((-\infty,t]) = \tilde{F}((-\infty,t])$ for all $t \in \R$. 
\end{proof}

\subsection{Proof for Proposition \ref{prop:delta_bound}}\label{sec:proof_of_prop2}
\begin{proof}
By definition $\Psi_{\bc^\top g,\bc^\top g}$, and using the fact that $E(\alpha)$ is a linear operator, we have,
 \begin{align}
        \Psi_{\bc^\top g,\bc^\top g}(\alpha) 
        = \langle E(\alpha) \bc^\top g,  \bc^\top g \rangle_{\pi}
        = \sum_{i=1}^d \sum_{j=1}^d \bc_i \bc_j \langle E(\alpha) g_i, g_j \rangle_\pi \nonumber.
    \end{align}
In particular, for any $a,b \in \R$ such that $a\le b$,
\begin{align*}
    \Psi_{\bc^\top g,\bc^\top g}(b) -\Psi_{\bc^\top g,\bc^\top g}(a) 
    &=\sum_{i=1}^d \sum_{j=1}^d \bc_i \bc_j \{\langle E(b) g_i, g_j \rangle_\pi-\langle E(a) g_i, g_j \rangle_\pi\}\\
    &=\sum_{i=1}^d \sum_{j=1}^d \bc_i \bc_j \langle (E(b)-E(a)) g_i, g_j \rangle_\pi\\
    &=\sum_{i=1}^d \sum_{j=1}^d \bc_i \bc_j \langle (E(b)-E(a)) g_i,  (E(b)-E(a)) g_j \rangle_\pi\\
    &\le \sum_{i=1}^d \sum_{j=1}^d |\bc_i \bc_j | \|  (E(b)-E(a)) g_i\|_{L^2(\pi)} \|(E(b)-E(a)) g_j \|_{L^2(\pi)}
\end{align*}
where for the third line we use the fact that $E(b)-E(a)$ is also a projection for $b\ge a$, and we use Cauchy-Schwartz inequality for the last inequality.

We note that for $f \in \mathcal{H}$ and $b\ge a$, $\|  (E(b)-E(a)) f\|_{L^2(\pi)}^2 = \Psi_{f,f}(b) - \Psi_{f,f}(a) $ since
\begin{align*}
    \|  (E(b)-E(a)) f\|_{L^2(\pi)}^2 = \langle (E(b)-E(a)) f, f\rangle_\pi = \langle E(b) f, f\rangle_\pi -\langle E(a) f, f\rangle_\pi .
\end{align*}
Also, since $\Psi_{\bc^\top g,\bc^\top g}$ is an increasing function, we have,
\begin{align}\label{lem:eq:Psi_cg_bound}
    0&\le\Psi_{\bc^\top g,\bc^\top g}(b) -\Psi_{\bc^\top g,\bc^\top g}(a) \nonumber\\
    &\le \sum_{i=1}^d \sum_{j=1}^d |\bc_i \bc_j | \sqrt{\Psi_{g_i,g_i}(b) - \Psi_{g_i,g_i}(a)}\sqrt{\Psi_{g_j,g_j}(b) - \Psi_{g_j,g_j}(a)}.
\end{align}
This implies that for any $(a,b] \subseteq \R$, if $F_{g_i}((a,b]) =\Psi_{g_i,g_i}(b) - \Psi_{g_i,g_i}(a) =0$ for all $i$ such that $\bc_i \ne 0$, then $F_{\bc^\top g} ((a,b]) = 0$.

Now suppose there exists $x$ such that $x \in {\rm Supp} (F_{\bc^\top g} ) \cap \{\cup_{i; \bc_i \ne 0} {\rm Supp}(F_{g_i})\}^c$. Since $x \in \{\cup_{i; \bc_i \ne 0} {\rm Supp}(F_{g_i})\}^c$ and $\{\cup_{i; \bc_i \ne 0} {\rm Supp}(F_{g_i})\}^c$ is open, we can find $\epsilon>0$ such that $(x-\epsilon, x+\epsilon) \subseteq \{\cup_{i; \bc_i \ne 0} {\rm Supp}(F_{g_i})\}^c$. Let $ A_\epsilon(x) = (x-\epsilon, x+\frac{\epsilon}{2}]$. We have $F_{g_i}(A_\epsilon(x)) =0$ for all $i$ such that $\bc_i\ne 0$ because for such $i$,
\begin{align*}
    F_{g_i}(A_\epsilon(x)) &= F_{g_i}(A_\epsilon(x) \cap  \{\cup_{i; \bc_i \ne 0} {\rm Supp}(F_{g_i})\}^c )\\
    &\le F_{g_i} (\cap_{i;\bc_i \ne 0} {\rm Supp} (F_{g_i})^c)\\
    &\le F_{g_i} ({\rm Supp} (F_{g_i})^c) = 0.
    \end{align*}
In other words, for any $i$ such that $\bc_i\ne 0$,
\begin{align*}
    F_{g_i}(A_\epsilon(x)) = \Psi_{g_i,g_i}(x+\frac{\epsilon}{2}) - \Psi_{g_i,g_i}(x-\epsilon) = 0.
\end{align*}
Now, since $x \in {\rm Supp}(F_{\bc^\top g})$, $F_{\bc^\top g} (A_\epsilon(x))\ge F_{\bc^\top g} ((x-\frac{\epsilon}{2}, x+\frac{\epsilon}{2}))>0$. 
On the other hand, using \eqref{lem:eq:Psi_cg_bound}, we obtain
    \begin{align*}
       F_{\bc^\top g}(A_\epsilon(x))
       &= \Psi_{\bc^\top g,\bc^\top g}(x+\frac{\epsilon}{2}) - \Psi_{\bc^\top g,\bc^\top g}(x-\epsilon)  \\
       &\le \left\lbrace \sum_{i=1}^d  |\bc_i  | \sqrt{\Psi_{g_i,g_i}(x+\frac{\epsilon}{2}) - \Psi_{g_i,g_i}(x-\epsilon)} \right\rbrace^2 =0
    \end{align*}
    which is a contradiction. Therefore, we conclude that ${\rm Supp} (F_{\bc^\top g} ) \cap \{\cup_{i; \bc_i \ne 0} {\rm Supp}(F_{g_i})\}^c = \emptyset$, i.e., ${\rm Supp} (F_{\bc^\top g} ) \subseteq  \cup_{i; \bc_i \ne 0} {\rm Supp}(F_{g_i})$.
\end{proof}

\section{Proofs for Section \ref{sec:statistical_guarantees}}

\subsection{Proof for Lemma~\ref{lem:stochastic_g_conv}\label{sec:proof_of_lemma1}}

\begin{proof}
First we note that for any $\bc \in \R^d$, $\bc^\top g = \sum_{i=1}^d \bc_i g_i$ is in $L_2(\pi)$ because
    \begin{align*}
        &\int \{\sum_{i=1}^d \bc_i g_i(x)\}^2 \pi (dx) \\
        &\le \sum_{i=1}^d \bc_i^2  \int g_i^2(x) \pi(dx) +2\sum_{1\le i<j\le d}|\bc_i \bc_j| \sqrt{\int  g_i^2(x) \pi(dx)  \int g_j^2 (x) \pi(dx) }
        <\infty
    \end{align*}
    where the last inequality follows from~\ref{cond:integrability}. Therefore, from Theorem \ref{thm:univariate_conv}, it suffices to show that $\rinitf{\bc_M^\top g}$ satisfies conditions \ref{cond:R1} - \ref{cond:R3}, for $\gamma=\gamma_{\bc^\top g}$ in~\ref{cond:R1}.
Obviously, $\rinitf{\bc_M^\top g}$ satisfies \ref{cond:R2} and \ref{cond:R3}.

In the below, we show that $\rinitf{\bc_M^\top g}$ satisfies \ref{cond:R1}, i.e., we show $\rinitf{\bc_M^\top g}(k) \to \gamma_{\bc^\top g}(k)$ $P_x$-almost surely for each $k\in \Z$, for each $x \in \mathsf{X}$. Let the initial condition $x \in \mathsf{X}$ be given.  
First, we have,
\begin{align*}
    \gamma_{\bc^\top g}(k) = {\rm Cov}_\pi (\bc^\top g(X_0), \bc^\top g(X_{k})) =\bc^\top {\rm Cov}_\pi ( g(X_0),  g(X_{k})) \bc = \sum_{i=1}^d \sum_{j=1}^d \bc_i \bc_j \gammaij(k)
\end{align*}
where we recall the definition $\gammaij(k) = {\rm Cov}_\pi ( g_i(X_t),  g_j(X_{t+k}))$.

We have,
\begin{align*}
    &\rinitf{\bc_M^\top g}(k) \\
    &=
    \frac{1}{M}\sum_{t=0}^{M-1-|k|} \sum_{i=1}^d\sum_{j=1}^d
    \bc_{Mi}\bc_{Mj} \tilde{g}_i(X_t) \tilde{g}_j(X_{t+|k|})\\
    &= \sum_{1\le i<j \le d}
    \bc_{Mi}\bc_{Mj}  \left \lbrace \frac{1}{M}\sum_{t=0}^{M-1-|k|}\{\tilde{g}_i(X_t) \tilde{g}_j(X_{t+|k|})+\tilde{g}_i(X_{t+|k|}) \tilde{g}_j(X_{t})\}\right\rbrace\\
    & \qquad + \sum_{i=1}^d \bc_{Mi}^2\left\lbrace  \frac{1}{M}\sum_{t=0}^{M-1-|k|}\tilde{g}_i(X_t) \tilde{g}_i(X_{t+|k|})\right\rbrace
\end{align*} 

By Lemma 8 in \cite{berg2022efficient}, for any $f\in L_2(\pi)$ we have
\[r_{f,M}(k) = \frac{1}{M}\sum_{t=0}^{M-1-|k|}\tilde{f}(X_t) \tilde{f}(X_{t+|k|}) \to \gamma_f(k)\]
$P_x$-almost surely.
Applying this Lemma to $f = g_i, f=g_j$, and $f=g_i+g_j$,
we have
\begin{align}\label{eq:modified_rij_conv}
    &r_{g_i+g_j,M}(k) - r_{g_i,M}(k) - r_{g_j,M}(k) \nonumber\\
    &=\frac{1}{M}\sum_{t=0}^{M-1-|k|}\{(\tilde{g}_i(X_t) +\tilde{g}_j(X_t))(\tilde{g}_i(X_{t+|k|})+\tilde{g}_j(X_{t+|k|})) - \tilde{g}_i(X_t) \tilde{g}_i(X_{t+|k|}) -\tilde{g}_j(X_t) \tilde{g}_j(X_{t+|k|}) \} \nonumber \\
    &=\frac{1}{M}\sum_{t=0}^{M-1-|k|}\{\tilde{g}_i(X_t) \tilde{g}_j(X_{t+|k|})+\tilde{g}_i(X_{t+|k|}) \tilde{g}_j(X_{t})\}\underset{M\to \infty}{\to} \gamma_{g_i +g_j}(k) -\gamma_{g_i}(k) - \gamma_{g_j}(k) 
\end{align}
$P_x$-almost surely. Note by definition of $\gammaij$, $\gamma_{g_i +g_j}(k) -\gamma_{g_i}(k) - \gamma_{g_j}(k)  = \gammaij(k) + \gammaji(k) $ for each $k\in \Z$. Therefore, w.p $1$,
\begin{align*}
    &\lim_{M\to\infty } \rinitf{\bc_M^\top g}(k)\\
    &= \sum_{1\le i<j \le d}
    \bc_{i}\bc_{j}  \left \lbrace  \lim_{M\to\infty }  \frac{1}{M}\sum_{t=0}^{M-1-|k|}\{\tilde{g}_i(X_t) \tilde{g}_j(X_{t+|k|})+\tilde{g}_i(X_{t+|k|}) \tilde{g}_j(X_{t})\}\right\rbrace\\
    & \qquad + \sum_{i=1}^d \bc_{i}^2\left\lbrace  \lim_{M\to\infty}\frac{1}{M}\sum_{t=0}^{M-1-|k|}\tilde{g}_i(X_t) \tilde{g}_i(X_{t+|k|})\right\rbrace\\
    &=\sum_{1\le i<j \le d}
    \bc_{i}\bc_{j} \{\gammaij(k) + \gammaji(k)\} +\sum_{i=1}^d \bc_i^2 \gamma_{g_i,g_i}(k)=\sum_{i=1}^d \sum_{j=1}^d \bc_i \bc_j \gammaij(k)=\gamma_{\bc^\top g}(k).
\end{align*} Thus $\rinitf{\bc_M^\top g}$ satisfies condition \ref{cond:R1} for $\gamma=\gamma_{\bc^\top g}$. Since~\ref{cond:R2}--\ref{cond:R3} also hold for $\rinitf{\bc_M^\top g}$, the Lemma follows from Theorem~\ref{thm:univariate_conv}. 
\end{proof}

\subsection{Proof for Theorem~\ref{thm:est1-conv}\label{sec:proof_of_Theorem2}}

\begin{proof}
        When $i=j$, this is a direct consequence of the $\ell_2$-convergence of the momentLS autocovariance sequence estimator for univariate functions of a chain. 
        Since each $w_{Mi} \to w_i \ne 0$, for a sufficiently large $M$, $w_{Mi} \ne 0$ for $i=1,\dots,d$. Therefore, without loss of generality, assume $w_{Mi}\ne 0$.
        When $i \ne j$, by \eqref{eq:gamma_ij_mseq}, $\gammaij(k) = \frac{1}{4w_iw_j}\{\gamma_{w_ig_i + w_jg_j}(k) - \gamma_{w_ig_i - w_jg_j}(k) \}$, and we have
  \begin{align*}
      &\|\hat{\gamma}^{(\boldsymbol{\delta})}_{M,ij}(\cdot; \mathbf{w}_M)- \gammaij(\cdot)\|_2  \\
      &=  \|\frac{1}{4w_{Mi} w_{Mj}}\{ \Pi_{\delta_{ij}}(r_{M,w_{Mi} g_i+w_{Mj} g_j})  - \Pi_{\delta_{ij}}(r_{M,w_{Mi} g_i+w_{Mj} g_j})\} - \frac{1}{4w_{i} w_{j}}(\gamma_{w_ig_i+w_jg_j} - \gamma_{w_ig_i-w_jg_j})\|_2\\
      &\leq I + II
  \end{align*}
  where 
  \begin{align*}
      I &= \|\frac{1}{4w_{Mi} w_{Mj}}\Pi_{\delta_{ij}}(r_{M,w_{Mi} g_i+w_{Mj} g_j}) - \frac{1}{4w_iw_j}\gamma_{w_ig_i+w_jg_j} \|_2\\
      II &= \|\frac{1}{4w_{Mi} w_{Mj}}\Pi_{\delta_{ij}}(r_{M,w_{Mi} g_i-w_{Mj} g_j}) - \frac{1}{4w_iw_j}\gamma_{w_ig_i-w_jg_j} \|_2
  \end{align*}
By the inequality
\begin{align}\label{eq:simple_ineq}
 \|a_M b_M - ab\|_2 
\le \|a_M(b_M-b)\|_2 + \|(a_M-a)b\|_2,
\end{align} for scalar $a_m,a \in \R$ and $b_M,b\in\ell_2(\mathbb{Z})$,
we have
\begin{align}
     &I= \|\frac{1}{4w_{Mi} w_{Mj}}\Pi_{\delta_{ij}}(r_{M,w_{Mi} g_i+w_{Mj} g_j}) - \frac{1}{4w_iw_j}\gamma_{w_ig_i+w_jg_j} \|_2\\
     &\le |\frac{1}{4w_{Mi} w_{Mj}} | \|\Pi_{\delta_{ij}}(r_{M,w_{Mi} g_i+w_{Mj} g_j}) -\gamma_{w_ig_i+w_jg_j} \|_2 + |\frac{1}{4w_{Mi} w_{Mj}}-\frac{1}{4w_iw_j}|\|\gamma_{w_ig_i+w_jg_j}\|_2\label{eq:termI}
\end{align}
Note $\delta_{ij} = \min\{\delta_i, \delta_j\} \le \Delta(F_{w_ig_i + w_jg_j})$ by Proposition \ref{prop:delta_bound}. Then, by Lemma \ref{lem:stochastic_g_conv} with $\bc_M = w_{Mi}  \mathbf{e}_i + w_{Mj} \mathbf{e}_j$ and $\bc= w_{i}  \mathbf{e}_i + w_{j} \mathbf{e}_j$ where $\mathbf{e}_j$ is the $j$th canonical basis vector, we have $\|\Pi_{\delta_{ij}}(r_{M,w_{Mi} g_i+w_{Mj} g_j}) -\gamma_{w_ig_i+w_jg_j} \|_2  \to 0$, $P_x$-a.s for any $x\in \mathsf{X}$. The second term in~\eqref{eq:termI} converges to 0 $P_x$-a.s. since $w_{Mi'}\overset{a.s.}{\to}w_{i'}\neq 0$ $P_x$ a.s. for each $i'=1,...,d$. Similarly, Term II also converges to $0$ $P_x$-a.s..

The second claim can be shown similarly as follows:
\begin{align*}
       &|\sum_{k\in\Z} \hat{\gamma}^{(\boldsymbol{\delta})}_{M,ij}(k;\mathbf{w}_M) -\sum_{k\in\Z} \gammaij(k)|\\
     &\le | \frac{1}{4w_{Mi}w_{Mj}} \sigma^2(\Pi_{\delta_{ij}}(r_{M,w_{Mi} g_i+w_{Mj} g_j})) -\frac{1}{4w_{i}w_{j}} \sigma^2(\gamma_{w_{i} g_i+w_{j} g_j}) |\\
     &\quad + | \frac{1}{4w_{Mi}w_{Mj}} \sigma^2(\Pi_{\delta_{ij}}(r_{M,w_{Mi} g_i-w_{Mj} g_j})) -\frac{1}{4w_{i}w_{j}} \sigma^2(\gamma_{w_{i} g_i-w_{j} g_j}) |
   \end{align*}
We see each term converges to $0$ by using a similar inequality as in \eqref{eq:simple_ineq} and applying the asymptotic variance convergence results in Theorem \ref{thm:univariate_conv} and Lemma \ref{lem:stochastic_g_conv}.
\end{proof}

\subsection{Proof for Lemma \ref{lem:sup_c_conv}}\label{sec:proof_of_Lemma2}
We prove the following lemma, which encompasses Lemma \ref{lem:sup_c_conv} as a part of its conclusion.
\begin{lem}
    Suppose the Markov chain $X$ satisfies \ref{cond:harris_ergodicity} - \ref{cond:geometric_ergodicity}. Let $g:\mathsf{X} \to \R^d$ be a function which satisfies \ref{cond:integrability}.
    Suppose $\delta$ is chosen so that $0<\delta\le  \dubar = \min_{1\le i \le d} \Delta(F_{g_i})$.
    We have
    \begin{enumerate}
        \item $\sup_{\bc; \|\bc\|_2 =1} \|\Pi_\delta(\rinitf{\mathbf{c}^\top g}) - \gamma_{\bc^\top g}\|_2 \to 0$, and
        \item $\sup_{\bc; \|\bc\|_2 =1}|\sigma^2(\Pi_\delta(\rinitf{\mathbf{c}^\top g}))- \sigma^2(\gamma_{\bc^\top g})|\to 0$
    \end{enumerate}
    $P_x$-almost surely for any $x \in \mathsf{X}$.    
\end{lem}
\begin{proof}
Let $\bc \in \R^d$ be given such that $\|\bc\|_2 = 1$. 
Let $\hat{F}_{\bc^\top g,\delta M}$ and $F_{\bc^\top g}$ be the representing measures for $\Pi_\delta(\rinitf{\bc^\top g})$ and $\gamma_{\bc^\top g}$ respectively. Note $\gamma_{\bc^\top g} \in \Mld{\delta}$ since ${\rm Supp}(F_{\bc^\top g}) \subseteq \cup_{i=1}^d {\rm Supp}(F_{g_i}) \subseteq [-1+\dubar, 1-\dubar] \subseteq [-1+\delta,1-\delta]$ by Proposition \ref{prop:delta_bound}.
Using Lemma 1 in \citet{berg2022efficient}, 
    \begin{align}
          &\sup_{\bc; \|\bc\|_2 =1}\|\Pi_\delta(\rinitf{\mathbf{c}^\top g}) - \gamma_{\bc^\top g}\|_2^2 \nonumber\\
          &\le \sup_{\bc; \|\bc\|_2 =1} \left \lbrace -\int\left\langle x_\alpha, \rinitf{\bc^\top g}-\gamma_{\bc^\top g}\right\rangle  F_{\bc^\top g}(d \alpha)+\int\left\langle x_\alpha,  \rinitf{\bc^\top g}-\gamma_{\bc^\top g}\right\rangle \hat{F}_{\bc^\top g,\delta M}(d \alpha)\right\rbrace\nonumber\\
          &\le\sup_{\bc; \|\bc\|_2 =1} \left\lbrace \sup_{\alpha \in [-1+\delta,1-\delta]} | \langle x_\alpha,  \rinitf{\bc^\top g}-\gamma_{\bc^\top g}\rangle|\right\rbrace \sup_{\bc; \|\bc\|_2 =1}  \{F_{\bc^\top g}([-1,1]) + \hat{F}_{\bc^\top g,\delta M}([-1,1])\}  \label{eq:lem4:bound1}
    \end{align}

We first show that 
\begin{align*}
\limsup_{M\to\infty} \sup_{\bc; \|\bc\|_2 =1} \left\lbrace \sup_{\alpha \in [-1+\delta,1-\delta]} | \langle x_\alpha,  \rinitf{\bc^\top g}-\gamma_{\bc^\top g}\rangle|\right\rbrace  = 0,\,\,\mbox{$P_x$- almost surely}    
\end{align*}

By definition of $r$,
\begin{align}
\rinitf{\bc^\top g}(k) 
&=\frac{1}{M}\sum_{t=0}^{M-1-|k|} \left(\sum_{i=1}^d \bc_i \tilde{g}_i(X_t )\right)\left(\sum_{j=1}^d \bc_j \tilde{g}_j(X_{t+|k|} )\right)  \nonumber   \\
&= \sum_{i=1}^d  \sum_{j=1}^d \bc_i \bc_j \frac{1}{M}\sum_{t=0}^{M-1-|k|}  \tilde{g}_i(X_t )  \tilde{g}_j(X_{t+|k|} )\nonumber \\
&= \sum_{i=1}^d \sum_{j=1}^d \bc_i \bc_j \tilde{r}_{ij,M}(k)\label{pf:lem2:eq:rij}
\end{align}
where we define
\begin{align*}
    \tilde{r}_{ij,M}(k) = 
    \begin{cases}
        \frac{1}{M}\sum_{t=0}^{M-1-|k|}\tilde{g}_i(X_t)  \tilde{g}_i(X_{t+|k|}) & i=j\\
        \frac{1}{2M}\sum_{t=0}^{M-1-|k|} \{\tilde{g}_i(X_t)  \tilde{g}_j(X_{t+|k|}) +\tilde{g}_j(X_t)  \tilde{g}_i(X_{t+|k|})\} & i\ne j.
    \end{cases}
\end{align*}

We also note from \eqref{eq:modified_rij_conv} that for $i\ne j$,
\[\tilde{r}_{ij,M}(k) = \{r_{g_i+g_j,M}(k) - r_{g_i,M}(k) - r_{g_j,M}(k)\}/2 \]
and in particular, $|\tilde{r}_{ij,M}(k)| \le 2^{-1} \{r_{g_i+g_j,M}(0) + r_{g_i,M}(0) + r_{g_j,M}(0)\}$ $P_x$-almost surely for all $1 \le i,j\le d$, since $\rinitf{g}(0)\geq |\rinitf{g}(k)|$ a.s., for all $k\geq 0$ for empirical autocovariance sequences of the form in~\eqref{eq:autocov} \citep[see, e.g.,][Proposition 7]{berg2022efficient}. We note \begin{align*}
    \gamma_{\bc^\top g}(k) = \operatorname{Cov}_\pi( (\bc^\top g)(X_0), (\bc^\top g)(X_k)) = \bc^\top  \operatorname{Cov}_\pi( g(X_0),  g(X_k)) \bc.\end{align*} Additionally, for $\gammaij(k) =  \operatorname{Cov}_\pi( g_i(X_0), g_j(X_k))$, we have \begin{align*}
    |\gammaij(k)| \le \sqrt{ \operatorname{Var}_\pi( g_i(X_0)) \operatorname{Var}_\pi (g_j(X_k)) }  = \sqrt{\gamma_{g_i}(0) \gamma_{g_j}(0)}.
    \end{align*}
For future reference to these bounds, we define
\begin{align*}
\check{r}_{ij,M} &= 0.5\{r_{g_i+g_j,M}(0) + r_{g_i,M}(0) + r_{g_j,M}(0)\}\\
\check{\gamma}_{ij} &= \sqrt{\gamma_{g_i}(0) \gamma_{g_j}(0)}.
\end{align*}
We note $\underset{M\to\infty}{\lim}\,\check{r}_{ij,M}=0.5\{\gamma_{g_i+g_j}(0)+\gamma_{g_i}(0)+\gamma_{g_j}(0)\}$ a.s. due to Assumption \ref{cond:integrability} for all $1\le i,j\le d$. In particular, $\underset{M}{\lim\sup}\,|\check{r}_{ij,M}|<\infty$ a.s. for $1\leq i,j\leq d$.

Coming back to bounding $\left\langle x_{\alpha},  \rinitf{\bc^\top g} -\gamma_{\bc^\top g} \right\rangle$, we have by \eqref{pf:lem2:eq:rij}
\begin{align*}
|\left\langle x_{\alpha},  \rinitf{\bc^\top g} -\gamma_{\bc^\top g} \right\rangle|
=|\sum_{k \in \Z}  \sum_{1 \le i,j \le d}  \bc_i \bc_j (\tilde{r}_{ij,M}(k)-\gammaij(k)) \alpha^{|k|}|
\end{align*}
We first show that the sequence $\{ \bc_i \bc_j (\tilde{r}_{ij,M}(k)-\gammaij(k)) \alpha^{|k|}\}$ is absolutely summable. Then we can exchange the order of summation by Fubini. We have
\begin{align*}
    &\sum_{k \in \Z}  \sum_{1 \le i,j \le d}  |\bc_i \bc_j (\tilde{r}_{ij,M}(k)-\gammaij(k)) \alpha^{|k|}|\\
    &\le \sum_{k \in \Z}  \sum_{1 \le i,j \le d}  |\bc_i \bc_j| |\tilde{r}_{ij,M}(k)-\gammaij(k)| |\alpha|^{|k|}\\
    &\le \sum_{k \in \Z}  \sum_{1 \le i,j \le d}  |\bc_i \bc_j|  |\alpha|^{|k|} \{ \check{r}_{ij,M}+\check{\gamma}_{ij}\}\\
    &=  \sum_{1 \le i,j \le d}  |\bc_i \bc_j|\{ \check{r}_{ij,M}+\check{\gamma}_{ij}\}  \sum_{k \in \Z}  |\alpha|^{|k|}\\
    &\le d \frac{1+|\alpha|}{1-|\alpha|}\max_{1\le i,j\le d}\{ \check{r}_{ij,M}+\check{\gamma}_{ij}\}<\infty.
\end{align*}
where for the final inequality we use the fact that $\|\bc\|_1^2 \le d \|\bc\|_2^2 = d$.
Therefore, by exchanging the summation order, we obtain
\begin{align*}
|\left\langle x_{\alpha},  \rinitf{\bc^\top g} -\gamma_{\bc^\top g} \right\rangle|
&=| \sum_{1 \le i,j \le d}  \bc_i \bc_j \sum_{k \in \Z} (\tilde{r}_{ij,M}(k)-\gammaij(k)) \alpha^{|k|}|\\
&\le \sum_{1 \le i,j \le d}  |\bc_i \bc_j| |\sum_{k \in \Z} (\tilde{r}_{ij,M}(k)-\gammaij(k)) \alpha^{|k|}|\\
&\le \left(\sum_{1 \le i,j \le d}  |\bc_i \bc_j|^2 \right)^{1/2}  \left(\sum_{1 \le i,j \le d}  |\langle \tilde{r}_{ij,M}-\gammaij,x_\alpha\rangle|^2\right)^{1/2}\\
& =  \left(\sum_{1 \le i,j \le d}  |\langle \tilde{r}_{ij,M}-\gammaij,x_\alpha\rangle|^2\right)^{1/2}
\end{align*}   
where for the second to last line we use Holder's inequality and for the last equality we use the fact that $\|\bc\|_2 = 1$.
Now
\begin{align*}
     &\sup_{\bc; \|\bc\|_2 =1}\left\lbrace\sup_{\alpha \in [-1+\delta,1-\delta]}|\left\langle x_{\alpha},  \rinitf{\bc^\top g} -\gamma_{\bc^\top g} \right\rangle| \right\rbrace\\
    &=\sup_{\alpha \in [-1+\delta,1-\delta]}\left\lbrace\sup_{\bc; \|\bc\|_2 =1}|\left\langle x_{\alpha},  \rinitf{\bc^\top g} -\gamma_{\bc^\top g} \right\rangle| \right\rbrace\\
     &\le \sup_{\alpha \in [-1+\delta,1-\delta]}\left(\sum_{1 \le i,j \le d}  |\langle \tilde{r}_{ij,M}-\gammaij,x_\alpha\rangle|^2\right)^{1/2}\\
     &\le \sup_{\alpha \in [-1+\delta,1-\delta]} \left( d^2 \max_{1\le i,j \le d} |\langle \tilde{r}_{ij,M}-\gammaij,x_\alpha\rangle|^2 \right)^{1/2}\\
      &= d \sup_{\alpha \in [-1+\delta,1-\delta]}  \max_{1\le i,j \le d} |\langle \tilde{r}_{ij,M}-\gammaij,x_\alpha\rangle|
\end{align*}
From \eqref{eq:modified_rij_conv}, we have that $\tilde{r}_{ij,M}(k) \to \gammaij(k)$ almost surely, for each $k\in \Z$. Following the same lines as in the proof of proposition 8 in \citet{berg2022efficient} except using the bounds $|\tilde{r}_{ij}(k)| \le \check{r}_{ij,M}$ and $|\gammaij(k)| \le \check{\gamma}_{ij}$, we obtain
\begin{align}
    \lim_{M\to \infty} \sup_{\alpha \in [-1+\delta,1-\delta]}   |\langle \tilde{r}_{ij,M}-\gammaij,x_\alpha\rangle| =0
\end{align} 
almost surely for each $1\le i,j\le d$, and thus 
\begin{align}\label{eq:lem4:bound2}
    \sup_{\bc; \|\bc\|_2 =1}\left\lbrace\sup_{\alpha \in [-1+\delta,1-\delta]}|\left\langle x_{\alpha},  \rinitf{\bc^\top g} -\gamma_{\bc^\top g} \right\rangle| \right\rbrace \le  d   \max_{1\le i,j \le d} \left(\sup_{\alpha \in [-1+\delta,1-\delta]}|\langle \tilde{r}_{ij,M}-\gammaij,x_\alpha\rangle|\right)\to 0
\end{align}
almost surely, since $d$ is finite and the order of supremums can be exchanged.

Now we show that the term
\[
 \sup_{\bc; \|\bc\|_2 =1} \left\lbrace F_{\bc^\top g}([-1,1]) + \hat{F}_{\bc^\top g,\delta M}([-1,1])\right\rbrace
\]
stays bounded.
By definition of the representing measure, 
\[F_{\bc^\top g}([-1,1]) = \gamma_{\bc^\top g}(0) = \bc^\top \operatorname{Cov}_\pi (g(X_0), g(X_0)) \bc=\sum_{1\le i,j \le d}  \bc_i\bc_j \gammaij(0)\]
and
\[\sum_{1\le i,j \le d}  \bc_i\bc_j \gammaij(0) \le \max_{i,j} \check{\gamma}_{ij}\sum_{1\le i,j \le d}  \bc_i\bc_j  \le \max_{i,j}\check{\gamma}_{ij} \|\bc\|_1^2 \le d \max_{i,j}\check{\gamma}_{ij}.\]
In particular,
$ \sup_{\bc; \|\bc\|_2 =1}   F_{\bc^\top g}([-1,1]) \le d \max_{i,j}\check{\gamma}_{ij} <\infty$. Thus, using the conclusion of Lemma \ref{lem:rep_Fhat_bounded} below, we have
\begin{align}\label{eq:lem4:bound3}
\limsup_{M\to \infty}\sup_{\bc ; \|\bc\|_2 =1} \{F_{\bc^\top g}([-1,1])+\hat{F}_{\bc^\top g,\delta M}([-1,1])\}
\le d( \max_{i,j}\check{\gamma}_{ij}) + C< \infty .
\end{align}
where $C$ is a constant defined in Lemma \ref{lem:rep_Fhat_bounded}.

Combining \eqref{eq:lem4:bound1}, \eqref{eq:lem4:bound2}, and \eqref{eq:lem4:bound3}, we conclude $\sup_{\bc; \|\bc\|_2 =1} \|\Pi_\delta(\rinitf{\mathbf{c}^\top g}) - \gamma_{\bc^\top g}\|_2 \to 0$, $P_x$-almost surely.

Having established the sequence convergence result, the asymptotic variance consistency result can be established following similar lines in the proof of Lemma 7 in \citet{berg2022efficient}, except we using the supremum of the total size of the representing measures
\[\sup_{\bc; \|\bc\|_2 =1} \left\lbrace F_{\bc^\top g}([-1,1]) + \hat{F}_{\bc^\top g,\delta M}([-1,1])\right\rbrace\]
for $B$ in their proof and using the supremum $\ell_2$ sequence convergence result, i.e.,
\[\sup_{\bc; \|\bc\|_2 =1} \|\Pi_\delta(\rinitf{\mathbf{c}^\top g}) - \gamma_{\bc^\top g}\|_2 \to 0.\]
\end{proof}

\begin{lem}\label{lem:rep_Fhat_bounded}
Suppose $\delta$ is chosen so that $0<\delta \le \underline{\delta} = \min_{1\le i \le d} \Delta(F_{g_i})$.
There exists a constant $C =  \frac{d(2-\delta)/\delta}{\inf_{\alpha, \alpha' \in [-1+\delta,1-\delta]} \left\langle x_{\alpha},  x_{\alpha'}\right\rangle}  \max_{1\le i,j\le d} \check{\gamma}_{ij} <\infty$, depending only on $\delta$ and $g$, such that 
\[\limsup_{M\to \infty}\sup_{\bc ; \|\bc\|_2 =1} \hat{F}_{\bc^\top g,\delta M}([-1,1])\le C,\] $P_x$- almost surely, where  $\hat{F}_{\bc^\top g,\delta M}$ is the representing measure for $\Pi_\delta(\rinitf{\bc^\top g})$ and $\check{\gamma}_{ij} = \sqrt{\gamma_{g_i}(0) \gamma_{g_j}(0)}$ for $1\le i,j \le d$.
\end{lem}
\begin{proof}
Let $\bc \in \R^d$ such that $\|\bc\|_2=1$ be given. We have $\operatorname{Supp}(\hat{F}_{\bc^\top g, \delta M})\subseteq [-1+\delta,1-\delta]  $ by the definition of the momentLS estimator $\Pi_\delta(\cdot)$.
Following the same lines as in the proof of Proposition 9 in \cite{berg2022efficient}, we obtain
\begin{align}
    \hat{F}_{\bc^\top g, \delta M}([-1,1])=\sum_{\alpha \in \operatorname{Supp}(\hat{F}_{\bc^\top g, \delta M})} \hat{F}_{\bc^\top g, \delta M}(\{\alpha\}) \le \frac{\sup_{\alpha \in [-1+\delta,1-\delta]} |\left\langle x_{\alpha},  \rinitf{\bc^\top g} \right\rangle|}{ \inf_{\alpha, \alpha' \in [-1+\delta,1-\delta]} \left\langle x_{\alpha},  x_{\alpha'} \right\rangle }\label{eq:sizeBound}
\end{align}
The denominator is deterministic and does not depend on $\bc$. Let $C_0:= \inf_{\alpha, \alpha' \in [-1+\delta,1-\delta]} \left\langle x_{\alpha},  x_{\alpha'} \right\rangle$. 
Now we show that for $\tilde{C}=\frac{d(2-\delta)}{\delta}\max_{1\leq i,j\leq d}\check{\gamma}_{ij}$, we have
\begin{align}\label{eq:lem4:innerproduct_bound}
    \limsup_{M\to\infty}\sup_{\bc ; \|\bc\|_2 =1} \sup_{\alpha \in [-1+\delta,1-\delta]} |\left\langle x_{\alpha},  \rinitf{\bc^\top g} \right\rangle| \le \tilde{C}
\end{align}
almost surely.
We have
\begin{align*} 
&\sup_{\bc ; \|\bc\|_2 =1} \sup_{\alpha \in [-1+\delta,1-\delta]} |\left\langle x_{\alpha},  \rinitf{\bc^\top g} \right\rangle| \\
 &\leq  \sup_{\bc ; \|\bc\|_2 =1}\sup_{\alpha \in [-1+\delta,1-\delta]} |\left\langle x_{\alpha},  \rinitf{\bc^\top g} -\gamma_{\bc^\top g} \right\rangle|+\sup_{\bc ; \|\bc\|_2 =1}\sup_{\alpha \in [-1+\delta,1-\delta]} |\left\langle x_{\alpha},  \gamma_{\bc^\top g} \right\rangle|.
\end{align*}
By \eqref{eq:lem4:bound2}, 
\[
\underset{M\to\infty}{\lim\sup}\; \sup_{\bc ; \|\bc\|_2 =1}\sup_{\alpha \in [-1+\delta,1-\delta]} |\left\langle x_{\alpha},  \rinitf{\bc^\top g} -\gamma_{\bc^\top g} \right\rangle| = 0
\]
almost surely. 
For the second term 
\begin{align*}
    |\left\langle x_{\alpha},  \gamma_{\bc^\top g} \right\rangle| 
    &=|\sum_{k \in \Z} \alpha^{|k|} \bc^\top \operatorname{Cov}_\pi(g(X_0),g(X_k)) \bc |\\
    &= |\sum_{k \in \Z}\sum_{1 \le i,j \le d} \alpha^{|k|} \bc_i \bc_j   \gammaij(k) |\\
    &\le \sum_{1 \le i,j \le d}|\bc_i \bc_j| \sum_{k \in \Z} |\alpha^{|k|}   \gammaij(k)  |\\
    &\le d  (\max_{1\le i,j\le d} \check{\gamma}_{ij})   \sum_{k \in \Z} |\alpha|^{|k|} 
\end{align*}
where for the second last inequality we use the Fubini-Tonelli Theorem and for the last inequality we use  $\|\bc\|_1 \le \sqrt{d} \|\bc\|_2 = \sqrt{d}$. 
Therefore,
\begin{align*}
    \sup_{\bc ; \|\bc\|_2 =1} \sup_{\alpha \in [-1+\delta,1-\delta]}  |\left\langle x_{\alpha},  \gamma_{\bc^\top g} \right\rangle| 
      \le  \frac{d(2-\delta)}{\delta}  \max_{1\le i,j\le d} \check{\gamma}_{ij}
\end{align*}
Therefore the inequality \eqref{eq:lem4:innerproduct_bound} holds with $\tilde{C} =  \frac{d(2-\delta)}{\delta}  \max_{1\le i,j\le d} \check{\gamma}_{ij}$. Thus, from~\eqref{eq:sizeBound}, we have \begin{align*}
    \underset{M\to\infty}{\lim\sup}\;\sup_{\bc ; \|\bc\|_2 =1}\;\hat{F}_{\bc^\top g, \delta M}([-1,1])\leq \tilde{C}/C_0=\frac{d(2-\delta)/\delta}{\inf_{\alpha, \alpha' \in [-1+\delta,1-\delta]} \left\langle x_{\alpha},  x_{\alpha'} \right\rangle}  \max_{1\le i,j\le d} \check{\gamma}_{ij}
\end{align*} almost surely.
\end{proof}

\subsection{Proof for Theorem~\ref{thm:mtv2}\label{sec:proof_of_theorem3}}
\begin{proof}
Let $\lambda_1\ge \dots\ge \lambda_d$ be eigenvalues of $\Sigma$ with corresponding eigenvectors $[\U_1,\dots,\U_d]$ such that $\Sigma \U_j = \lambda_j\U_j$ for $j=1,\dots,d$. 
Also, recall the definition of $\hat{\Sigma}^{(\boldsymbol{\delta})}_{psd} = \hat{\U}_M \hat{\bLambda}_{\boldsymbol{\delta}} \hat{\U}_M^\top$, where  $\hat{\bLambda}_{\boldsymbol{\delta}}$ is a diagonal matrix with $j$th diagonal entry $\sigma^2(\Pi_{
\underline{\delta}}(\rinitf{\hat{\U}_{Mj}^\top g}))$ and $\hat{\U}_M = [\hat{\U}_{M1},\dots,\hat{\U}_{Md}]$ is the collection of eigenvectors for the original estimator $\hat{\Sigma}^{(\boldsymbol{\delta})}_{pw}$, such that $\hat{\Sigma}^{(\boldsymbol{\delta})}_{pw} \hat{\U}_{Mj} = \hat{\lambda}^{pw}_j\hat{\U}_{Mj}$ for $j=1,\dots,d$ and $\hat{\lambda}^{pw}_1 \ge \dots \geq \hat{\lambda}^{pw}_d$.
For notational simplicity, we denote $\hat{\lambda}_j = \sigma^2(\Pi_{
\dubar}(\rinitf{\hat{\U}_{Mj}^\top g}))$, $j=1,\dots,d$ and $\tilde{\Sigma} = \hat{\Sigma}^{(\boldsymbol{\delta})}_{psd}$. Since $\tilde{\Sigma} \hat{\U}_M = \hat{\U}_M \hat{\bLambda}_{\boldsymbol{\delta}}$, each $(\hat{\lambda}_j, \hat{\U}_{Mj})$ is an eigenpair of $\tilde{\Sigma}$ for $j=1,\dots,d$. 

Let $r$ be the number of distinct eigenvalues in $(\lambda_1,\dots,\lambda_d)$ and $\lambda_{(1)}> \dots > \lambda_{(r)}\ge 0$ be the ordered distinct eigenvalues. Define $S_j= \{k; \lambda_k = \lambda_{(j)} \}$ and let $\U_{S_j}$ and $\hat{\U}_{MS_j}$ be sub-matrices consisting of $\{\U_k\}_{k\in S_j}$ and $\{\hat{\U}_{Mk}\}_{k\in S_j}$, respectively, for $j=1,\dots,r$.

We have,
\begin{align}
    \| \tilde{\Sigma}- \Sigma \|_{F}
     &= \| \sum_{j=1}^d \hat{\lambda}_j \hat{\U}_{Mj}\hat{\U}_{Mj}^\top - \sum_{j=1}^d \lambda_j \U_j \U_j^\top\|_F \nonumber\\
     &= \| \sum_{j=1}^d (\hat{\lambda}_j -\lambda_j) \hat{\U}_{Mj} \hat{\U}_{Mj}^\top + \sum_{j=1}^d \lambda_j (\hat{\U}_{Mj} \hat{\U}_{Mj}^\top -  \U_j \U_j^\top) \|_{F}\nonumber\\
     &\le \sum_{j=1}^d |\hat{\lambda}_j -\lambda_j| +\sum_{j=1}^r \lambda_{(j)}  \| \hat{\U}_{MS_j} \hat{\U}_{MS_j}^\top -  \U_{S_j} \U_{S_j}^\top\|_F \label{eq:thm2:sigma_bound}
\end{align}
where for the last inequality we use the fact $\|\hat{\U}_{Mj} \hat{\U}_{Mj}^\top \|_F = \sqrt{\textrm{tr}(\hat{\U}_{Mj} \hat{\U}_{Mj}^\top\hat{\U}_{Mj} \hat{\U}_{Mj}^\top)} = 1$ for $j=1,\dots,d$.

For the second term in \eqref{eq:thm2:sigma_bound},
\begin{align*}
    \| \hat{\U}_{MS_j} \hat{\U}_{MS_j}^\top -  \U_{S_j} \U_{S_j}^\top\|_F 
    =\frac{1}{\sqrt{2}}\|\textrm{Sin}  \,\Theta (\U_{S_j}, \hat{\U}_{MS_j} )\|_F^2 .
\end{align*}
where $\Theta(\U_{S_j}, \hat{\U}_{MS_j} )$ denotes the $|S_j|$ by $|S_j|$ diagonal matrix whose diagonal entries are the principal angles of the column spaces generated by $\U_{S_j}$ and $\hat{\U}_{MS_j}$ respectively, and $\sin \Theta(\U_{S_j}, \hat{\U}_{MS_j} )$ is defined entry-wise. More specifically, \[\sin \Theta (\U_{S_j}, \hat{\U}_{MS_j} ) = {\rm diag}\{\sin (\cos^{-1}(\sigma_1)),\dots,\sin(\cos^{-1}(\sigma_{|S_j|}))\}\] where $\sigma_1\ge \dots\ge\sigma_{|S_j|}$ are the singular values of $\hat{\U}_{MS_j}^\top \U_{S_j}$.

By a variant of the Davis-Kahan theorem (Theorem 2 in \cite{yu2015useful}), for each $j=1,\dots,r$, 
$$
\| \sin\Theta (\U_{S_j}, \hat{\U}_{MS_j} )\|_F
\leq \frac{2 \min \left(|S_j|^{1 / 2}\|\hat{\Sigma}^{(\boldsymbol{\delta})}_{pw}-\Sigma\|_{\mathrm{op}},\|\hat{\Sigma}^{(\boldsymbol{\delta})}_{pw}-\Sigma\|_{\mathrm{F}}\right)}{\min \left(\lambda_{(j-1)}-\lambda_{(j)}, \lambda_{(j)}-\lambda_{(j+1)}\right)} .
$$
where we define $\lambda_{(0)} = \infty$ and $\lambda_{(r+1)} = -\infty$. In particular, since $d$ is fixed and $\hat{\Sigma}^{(\boldsymbol{\delta})}_{pw, ij} \to \Sigma_{ij}$ for each $i,j$, $P_x$-a.s., we have $\| \sin\Theta (\U_{S_j}, \hat{\U}_{MS_j} )\|_F \to 0$ for each $j$, $P_x$-a.s.

For the first term in \eqref{eq:thm2:sigma_bound}, we have for each $i=1,\dots,d$,
\begin{align}
    |\hat{\lambda}_i - \lambda_i| 
    &\le |\sigma^2(\Pi_{
\dubar}(\rinitf{\hat{\U}_{Mi}^\top g})) - \sigma^2 (\gamma_{\hat{\U}_{Mi}^\top g}) | + |\sigma^2 (\gamma_{\hat{\U}_{Mi}^\top g}) - \sigma^2 (\gamma_{\U_i^\top g})| \nonumber\\
&\le \sup_{\bc; \|\bc\|_2 =1} |\sigma^2(\Pi_{
\dubar}(\rinitf{\bc^\top g})) - \sigma^2 (\gamma_{\bc^\top g}) | + |\sigma^2 (\gamma_{\hat{\U}_{Mi}^\top g}) - \sigma^2 (\gamma_{\U_i^\top g})| \label{eq:thm2:eigen_bound}
\end{align}
The first term in \eqref{eq:thm2:eigen_bound} converges to $0$ almost surely due to Lemma \ref{lem:sup_c_conv} since $\dubar = \min_{1\le i\le d} \delta_i \le \min_{1 \le i \le d} \Delta(F_{g_i})$ by the condition of $\boldsymbol{\delta}=[\delta_1,\dots,\delta_d]$ in Theorem \ref{thm:est1-conv}. For the second term in \eqref{eq:thm2:eigen_bound}, for each $\bc \in \R^d$, we have $\sigma^2 (\gamma_{\bc^\top g}) = 
 \bc^\top \Sigma \bc$ and
\begin{align*}
    |\sigma^2 (\gamma_{\hat{\U}_{Mi}^\top g}) - \sigma^2 (\gamma_{\U_i^\top g})| = | \hat{\U}_{Mi}^\top \Sigma \hat{\U}_{Mi} - \U_i^\top \Sigma \U_i|
\end{align*}
Let $j$ be the index such that $\lambda_i = \lambda_{(j)}$. That is $i \in S_j = \{k; \lambda_k = \lambda_{(j)}\}$. We note that for any orthogonal matrix $\mathbf{O} \in \R^{|S_j| \times |S_j|}$, a column of the matrix $\U_{S_j} \mathbf{O} \in \R^{d \times |S_j|}$ is also an eigenvector corresponding to $\lambda_{(j)}$ since for $l = 1,\dots, |S_j|$
\begin{align*}
    \Sigma (\U_{S_j} \mathbf{O})_l = \Sigma \left (\sum_{k=1}^{|S_j|} \U_{S_j,k} \mathbf{O}_{kl}\right) =   \sum_{k=1}^{|S_j|} \Sigma\U_{S_j,k} \mathbf{O}_{kl} = \lambda_{(j)} \sum_{k=1}^{|S_j|} \U_{S_j,k} \mathbf{O}_{kl} = \lambda_{(j)} (\U_{S_j} \mathbf{O})_l
\end{align*}
Moreover, by Theorem 2 in \cite{yu2015useful}, for each $j=1,\dots,r$ there exists an orthogonal matrix $\hat{\mathbf{O}}_{Mj} \in \R^{|S_j| \times |S_j|}$ such that
\begin{align}\label{eq:DK-2}
\|\U_{S_j} - \hat{\U}_{MS_j}\hat{\mathbf{O}}_{Mj} \|_F
\leq \frac{2^{3/2} \min \left(|S_j|^{1 / 2}\|\hat{\Sigma}^{(\boldsymbol{\delta})}_{pw}-\Sigma\|_{\mathrm{op}},\|\hat{\Sigma}^{(\boldsymbol{\delta})}_{pw}-\Sigma\|_{\mathrm{F}}\right)}{\min \left(\lambda_{(j-1)}-\lambda_{(j)}, \lambda_{(j)}-\lambda_{(j+1)}\right)} .
\end{align}
Let us define $\tilde{\mathbf{O}}_{Mj} = \hat{\mathbf{O}}_{Mj}^{-1}$, and let $l$ be the column index for $\hat{\U}_{Mi}$ in $\hat{\U}_{MS_j}$, i.e., $\hat{\U}_{MS_j} \mathbf{e}_l = \hat{\U}_{Mi}$. Note that $\tilde{\mathbf{O}}_{Mj}$ is also an orthogonal matrix, and therefore $ (\U_{S_j}\tilde{\mathbf{O}}_{Mj})\mathbf{e}_l$ is an eigenvector of $\Sigma$ whose corresponding eigenvalue is $\lambda_{(j)} = \lambda_i$. In particular, $\lambda_i = \U_i^\top \Sigma \U_i = \mathbf{e}_l^\top (\U_{S_j}\tilde{\mathbf{O}}_{Mj})^\top\Sigma (\U_{S_j}\tilde{\mathbf{O}}_{Mj})\mathbf{e}_l$. 
Therefore, we have,
\begin{align*}
    |\sigma^2 (\gamma_{\tilde{\U}_{i}^\top g}) - \sigma^2 (\gamma_{\U_i^\top g})| &= | \hat{\U}_{Mi}^\top \Sigma \hat{\U}_{Mi} - \U_i^\top \Sigma \U_i| \\
   &=| \mathbf{e}_l^\top \hat{\U}_{MS_j}^\top \Sigma \hat{\U}_{MS_j}\mathbf{e}_l - \mathbf{e}_l^\top (\U_{S_j}\tilde{\mathbf{O}}_{Mj}) ^\top \Sigma (\U_{S_j}\tilde{\mathbf{O}}_{Mj})\mathbf{e}_l|\\
   &\le \|\hat{\U}_{MS_j}\mathbf{e}_l  - (\U_{S_j}\tilde{\mathbf{O}}_{Mj})\mathbf{e}_l\|_2\{\|\Sigma \hat{\U}_{MS_j}\mathbf{e}_l\|_2+ \|\Sigma (\U_{S_j}\tilde{\mathbf{O}}_{Mj})\mathbf{e}_l\|_2\}\\
   &\le 2 \|\hat{\U}_{MS_j}  - (\U_{S_j}\tilde{\mathbf{O}}_{Mj})\|_F \|\Sigma\|_F
\end{align*}
where the first inequality uses the fact that 
\begin{align*}
    |\mathbf{a}^\top \Sigma \mathbf{a} - \mathbf{b}^\top \Sigma \mathbf{b}| \le \|\mathbf{a}- \mathbf{b}\|_2 \{ \|\Sigma \mathbf{a}\|_2 + \|\Sigma \mathbf{b}\|_2\}
\end{align*}
and the second inequality uses the fact that for any $\mathbf{A} \in \R^{d \times r}$ and $\mathbf{b} \in \R^{r}$,
$\|\mathbf{A}\mathbf{b} \|_2 \le \|\mathbf{A}\|_F \|\mathbf{b}\|_2$.

Now, 
\begin{align*}
    \|\hat{\U}_{MS_j}  - (\U_{S_j}\tilde{\mathbf{O}}_{Mj})\|_F = \|\hat{\U}_{MS_j}\hat{\mathbf{O}}_{Mj}  - (\U_{S_j})\|_F \underset{M\to \infty}{\to} 0
\end{align*}
$P_x$-almost surely by \eqref{eq:DK-2}, which shows that $|\hat{\lambda}_j - \lambda_j| \underset{M\to \infty}{\to} 0$ $P_x$-almost surely for each $j=1,\dots,d$. This implies $\sum_{j=1}^d |\hat{\lambda}_j - \lambda_j| \underset{M\to \infty}{\to} 0$, $P_x$-a.s., since $d$ is finite.
\end{proof}

\newpage
\section{Algorithm for MomentLS estimators}\label{sec:alg_momentLS}

We summarize the steps to compute the multivariate moment LS asymptotic variance estimator (mtv-mLSE) $\hat{\Sigma}^{(\boldsymbol{\delta})}$ in Algorithm \ref{alg:mtvMomentLS}. Algorithm \ref{alg:mtvMomentLS} invokes functions related to asymptotic variance estimation via the univariate momentLS method, which we summarize in Algorithm \ref{alg:uMomentLS}.

\begin{algorithm}[H]
\setstretch{1}
\SetAlgoLined
\SetAlgoNoEnd
\SetKwProg{Prog}{Function}{}{}
\SetKwFunction{fmtvpw}{mtvMomentLS}
\Prog{\fmtvpw{chain $\mathbf{Y} = [\mathbf{Y}_1,\dots, \mathbf{Y}_d]\in \R^{M\times d}$, batch number for delta tuning $L$, initial grid size $s$}}{
$\Theta_0  = \{\alpha_1,\dots,\alpha_s\}\subseteq [-1,1]$ \tcp*{initial length-$s$ grid in $[-1,1]$}
\tcp{estimate diagonal elements of $\Sigma$}
\For{$i=1,\dots,d$}{
$r_i \gets$ \texttt{emp\_autocov}($\mathbf{Y}_i$)\tcp*{empirical autocov sequence of each $\mathbf{Y}_i$}
\tcp{momentLS asymptotic variance for $\mathbf{Y}_i$}
$\delta_i \gets$ \texttt{tune\_delta}($\mathbf{Y}_i,L$)\;
$\Sigma_{ii} \gets$ \texttt{asympVar\_u\_momentLS}($r_i, \delta_i,\Theta_0$)\;
}
\tcp{estimate off-diagonal elements of $\Sigma$}
\For{$i=1,\dots,d-1$}{
\For{$j=(i+1),\dots,d$}
{
$\delta_{ij} \gets \min\{\delta_i,\delta_j\}$\; 
$a \gets 1/\sqrt{r_i(0)}$; $b \gets 1/\sqrt{r_j(0)}$\ \tcp*{scale of $\mathbf{Y}_i,\mathbf{Y}_j$}
\tcp{empirical autocovariance for $a\mathbf{Y}_i+b\mathbf{Y}_j$, $a\mathbf{Y}_i-b\mathbf{Y}_j$}
$r_{ij,1} \gets$  \texttt{emp\_autocov}($a\mathbf{Y}_i+b\mathbf{Y}_j$)\; $r_{ij,2} \gets$  \texttt{emp\_autocov}($a\mathbf{Y}_i-b\mathbf{Y}_j$)\; 
\tcp{momentLS asymptotic variance for $a\mathbf{Y}_i+b\mathbf{Y}_j$, $a\mathbf{Y}_i-b\mathbf{Y}_j$}
$v_{ij,1} \gets$\texttt{asympVar\_u\_momentLS}($r_{ij,1}, \delta_{ij},\Theta_0$)\;
$v_{ij,2} \gets$ \texttt{asympVar\_u\_momentLS}($r_{ij,2}, \delta_{ij},\Theta_0$)\;
$\Sigma_{ij} \gets (v_{ij,1} -v_{ij,2})/(4ab)$; $\Sigma_{ji} \gets \Sigma_{ij}$ 
}}
\tcp{if $\Sigma$ is not psd, compute psd refinement of $\Sigma$}
$(\mathbf{U},\mathbf{\Lambda}) \gets$ \texttt{eigen}($\Sigma$)\tcp*{eigenvalue decomposition to obtain $\mathbf{U}\in \R^{d\times d}$= eigenvectors, $\mathbf{\Lambda} \in \R^d$= eigenvalues of $\Sigma$}
\If{$\min(\Lambda) <0$ }{
$\mathbf{Y}_D \gets \mathbf{Y}\mathbf{U}$\;
\For{$i = 1,\dots,d$}{
$r_{Di} \gets$ \texttt{emp\_autocov}($(\mathbf{Y}_D)_i$)\;
$\lambda_i \gets$ \texttt{asympVar\_u\_momentLS}($r_{Di},\min_{1\le i \le d}\delta_i,\Theta_0$)
}
$\Sigma  \gets$ $\mathbf{U}\mbox{diag}(\lambda_1,\dots,\lambda_d) \mathbf{U}^\top$
}
\KwRet  $\Sigma$
}
\caption{multivariate MomentLS estimator}\label{alg:mtvMomentLS}
\end{algorithm}

\begin{algorithm}
\SetAlgoNoEnd
  \setstretch{1}
  \SetKwProg{Prog}{Function}{}{}
  \SetKwFunction{fempautocov}{emp\_autocov}
  \Prog{\fempautocov{chain $Y \in \R^M$}}{

  $\tilde{Y} = Y - ( M^{-1}\sum_{t=0}^{M-1} Y_{t})\mathbf{1}_M$ \tcp{centered $Y_t$}
  
  \For{$k=0,\dots,M-1$}{$r(k) \gets M^{-1}\sum_{t=0}^{M-1} \tilde{Y}_t\tilde{Y}_{t+k}$ \tcp{lag $k$ empirical autocov}}
  
  \KwRet empirical autocovariance sequence $r$\;
  }
  \SetKwFunction{ftuned}{tune\_delta}
  \Prog{\ftuned{chain $Y \in \R^M$,batch number $L$} }{
   $\tilde{Y} = Y - ( M^{-1}\sum_{t=0}^{M-1} Y_{t})\mathbf{1}_M$ \;
  $B\gets \lfloor M / L\rfloor$\;
  \tcp{compute autocov sequence for each split of $Y$}
  \For{l = $1,...,L$}{
  \For{$k = 1,\dots,B-1$}{
 \If{$l=1$}{$\tilde{r}^{(l)}(k) \gets \frac{1}{B} \sum_{t=0}^{B-1-k} \tilde{Y}_t\tilde{Y}_{t+k}$}
 \Else{
 $\tilde{r}^{(l)}(k) \gets\frac{1}{B} \sum_{t=(l-1) B-k}^{l B-1-k} \tilde{Y}_t\tilde{Y}_{t+k}$
 }
  }}
  \tcp{estimate delta value for each split}
\For{$l=1,\dots,L$}{
$m^{(l)}\gets \min \left\{t \in 2 \mathbb{N} ; \tilde{r}^{(l)}(t+2) \leq 0\right\}$\;
$\delta^{(l)}=\max \{1-\exp \{-\log (B) /(2 m^{(l)})\}, 1 / B\}$
}
  \KwRet 0.8$\times$(average delta): $0.8(L^{-1}\sum_{l=1}^L \delta^{(l)})$
  }
  \SetKwFunction{fumomentLS}{asympVar\_u\_momentLS}
  \Prog{\fumomentLS{input autocov sequence $r$, tuning param $\delta$, initial grid $\Theta_0$, grid length $s_0=1001$}}{
  \If{initial grid $\Theta_0$ is not \texttt{NULL}}{
  obtain $\Theta \subseteq \Theta_0$ such that $\Theta \subseteq [-1+\delta, 1-\delta]$
  }
  \Else{create a length $s_0$ grid $\Theta$ in $[-1+\delta,1-\delta]$}
  $s \gets |\Theta|$ \tcp*{size of grid $\Theta$}
  \tcp{solve a constrained optimization problem}
  \For{$i=1,\dots,s$}{
  $\alpha_i \gets \Theta[i]$ \tcp*{$i$th element of grid $\Theta$}
  $\mathbf{a}_i \gets \sum_{k ; r(k) \neq 0} \alpha_i^{|k|} r(k)$\;
  \For{$j=1,\dots,s$}{
  $\mathbf{B}_{i j}\gets \frac{1+\alpha_i \alpha_j}{1-\alpha_i \alpha_j}$
  }
  }

  $\hat{\mathbf{w}} \gets \arg\min _{\mathbf{w};  \mathbf{w} \geq 0} r^{\top} r-2 \mathbf{a}^{\top} \mathbf{w}+\mathbf{w}^{\top} \mathbf{B} \mathbf{w}$\;
  \tcp{support and weight of representing measure of moment LS estimator}
  $S \gets \{\alpha_i\in \Theta; \hat{\mathbf{w}}_i>0\}$; $\hat{\mathbf{w}} \gets \{\hat{\mathbf{w}}_i; \hat{\mathbf{w}}_i>0\}$ \;
  \tcp{compute asymptotic variance from $(S,\hat{\mathbf{w}})$}
  $\sigma^2 \gets \sum_{i=1 }^{|S|} \hat{\mathbf{w}}_i\frac{1+\alpha_i}{1-\alpha_i}$\;
  \KwRet $\sigma^2$
  }
  
  \caption{functions related to univariate moment LS estimator}\label{alg:uMomentLS}
\end{algorithm}

\newpage

\section{Ground Truth Values}\label{sec:oracle_values}


\paragraph{Estimated posterior means:}
\,
\begin{table}[H]
\centering
\small
\begin{tabular}{lcccccc}
  \hline
Sampler & Intercept & mcv & alkphos & sgpt & sgot & gammagt \\ 
  \hline
RW-metrop & -0.200446 & 0.497093 & -0.062117 & -0.375257 & 0.567321 & 0.418176 \\ 
  NUTS & -0.200386 & 0.497100 & -0.062089 & -0.375296 & 0.567382 & 0.418184 \\ 
  PG & -0.200385 & 0.497096 & -0.062095 & -0.375287 & 0.567365 & 0.418167 \\ 
   \hline
\end{tabular}
\caption{Estimated ground truth posterior mean of each coefficient for the Bayesian logistic regression example for the three samplers from $10000$ independent chains of length $40000$ (after burn-in periods of length $5000$). Each sampler targets the same posterior, so the true posterior means are identical.}\label{tab:trueMean}
\end{table}

\paragraph{Estimated ground truth asymptotic variance matrices:}
\,
\begin{table}[H]
\centering
\small
\setlength{\tabcolsep}{2pt}
\begin{subtable}{.475\textwidth}
\centering
\caption{RW-metrop}
\begin{tabular}{c|cccccc}
  \hline
 & 1 & 2 & 3 & 4 & 5 & 6 \\ 
   \hline
1 & 0.274 & -0.004 & -0.004 & -0.009 & 0.016 & 0.037 \\ 
  2 & -0.004 & 0.341 & -0.007 & -0.028 & 0.015 & -0.025 \\ 
  3 & -0.004 & -0.007 & 0.282 & 0.016 & -0.048 & -0.036 \\ 
  4 & -0.009 & -0.028 & 0.016 & 0.773 & -0.497 & -0.205 \\ 
  5 & 0.016 & 0.015 & -0.048 & -0.497 & 0.797 & -0.098 \\ 
  6 & 0.037 & -0.025 & -0.036 & -0.205 & -0.098 & 0.598 \\ 
   \hline
\end{tabular}
\end{subtable}
\begin{subtable}{.475\textwidth}\caption{RW-metrop (one chain)}
\centering
\begin{tabular}{c|cccccc}
  \hline
  & 1 & 2 & 3 & 4 & 5 & 6 \\ 
  \hline
1 & 0.278 & -0.020 & -0.002 & 0.027 & -0.002 & 0.043 \\ 
  2 & -0.020 & 0.390 & -0.007 & -0.030 & 0.061 & -0.047 \\ 
  3 & -0.002 & -0.007 & 0.267 & 0.005 & -0.048 & -0.025 \\ 
  4 & 0.027 & -0.030 & 0.005 & 0.725 & -0.501 & -0.145 \\ 
  5 & -0.002 & 0.061 & -0.048 & -0.501 & 0.805 & -0.084 \\ 
  6 & 0.043 & -0.047 & -0.025 & -0.145 & -0.084 & 0.544 \\ 
   \hline
\end{tabular}
\end{subtable}

\begin{subtable}{.475\textwidth}
\centering
\caption{NUTS}
\begin{tabular}{c|cccccc}
  \hline
 & 1 & 2 & 3 & 4 & 5 & 6 \\ 
  \hline
1 & 0.012 & -0.001 & -0.000 & -0.003 & 0.002 & 0.003 \\ 
  2 & -0.001 & 0.014 & -0.000 & -0.001 & 0.002 & -0.002 \\ 
  3 & -0.000 & -0.000 & 0.011 & 0.004 & -0.005 & -0.003 \\ 
  4 & -0.003 & -0.001 & 0.004 & 0.056 & -0.045 & -0.014 \\ 
  5 & 0.002 & 0.002 & -0.005 & -0.045 & 0.056 & 0.001 \\ 
  6 & 0.003 & -0.002 & -0.003 & -0.014 & 0.001 & 0.028 \\ 
   \hline
\end{tabular}
\end{subtable}
\begin{subtable}{.475\textwidth}\caption{NUTS (one chain)}
\centering
\begin{tabular}{c|cccccc}
  \hline
 & 1 & 2 & 3 & 4 & 5 & 6 \\ 
  \hline
1 & 0.011 & -0.001 & -0.000 & -0.003 & 0.002 & 0.003 \\ 
  2 & -0.001 & 0.014 & -0.000 & -0.001 & 0.002 & -0.003 \\ 
  3 & -0.000 & -0.000 & 0.012 & 0.004 & -0.005 & -0.003 \\ 
  4 & -0.003 & -0.001 & 0.004 & 0.055 & -0.044 & -0.015 \\ 
  5 & 0.002 & 0.002 & -0.005 & -0.044 & 0.055 & 0.001 \\ 
  6 & 0.003 & -0.003 & -0.003 & -0.015 & 0.001 & 0.029 \\ 
   \hline
\end{tabular}
\end{subtable}
\begin{subtable}{.475\textwidth}
\centering
\caption{PG}
\begin{tabular}{c|cccccc}
  \hline
 & 1 & 2 & 3 & 4 & 5 & 6 \\ 
  \hline
1 & 0.017 & -0.001 & -0.000 & -0.002 & 0.002 & 0.006 \\ 
  2 & -0.001 & 0.024 & 0.000 & -0.003 & 0.005 & 0.000 \\ 
  3 & -0.000 & 0.000 & 0.016 & 0.002 & -0.003 & -0.003 \\ 
  4 & -0.002 & -0.003 & 0.002 & 0.055 & -0.037 & -0.019 \\ 
  5 & 0.002 & 0.005 & -0.003 & -0.037 & 0.058 & -0.001 \\ 
  6 & 0.006 & 0.000 & -0.003 & -0.019 & -0.001 & 0.050 \\ 
   \hline
\end{tabular}
\end{subtable}
\begin{subtable}{.475\textwidth}\caption{PG (one chain)}
\centering
\begin{tabular}{c|cccccc}
  \hline
 & 1 & 2 & 3 & 4 & 5 & 6 \\ 
   \hline
1 & 0.017 & -0.001 & -0.000 & -0.001 & 0.002 & 0.006 \\ 
  2 & -0.001 & 0.024 & 0.000 & -0.003 & 0.005 & 0.000 \\ 
  3 & -0.000 & 0.000 & 0.016 & 0.001 & -0.003 & -0.003 \\ 
  4 & -0.001 & -0.003 & 0.001 & 0.055 & -0.037 & -0.019 \\ 
  5 & 0.002 & 0.005 & -0.003 & -0.037 & 0.058 & -0.002 \\ 
  6 & 0.006 & 0.000 & -0.003 & -0.019 & -0.002 & 0.050 \\ 
   \hline
\end{tabular}
\end{subtable}
\caption{Tables (a), (c), and (e) show estimated ground truth asymptotic variance matrices for the three samplers from $10000$ independent chains of length $40000$ (after $5000$ burn-in). Tables (b), (d), and (f) show a single realization of the momentLS estimate from a length $M=40000$ sample. Each chain targets the same posterior. However, the different sampling mechanisms lead to different autocovariance sequences, and hence different values of $\Sigma$.}\label{tab:avarComp}

\end{table}

\iftwofiles\putbib
\end{bibunit}\else
\bibliographystyle{plainnat}
\bibliography{bib}

\begin{thebibliography}{35}
\providecommand{\natexlab}[1]{#1}
\providecommand{\url}[1]{\texttt{#1}}
\expandafter\ifx\csname urlstyle\endcsname\relax
  \providecommand{\doi}[1]{doi: #1}\else
  \providecommand{\doi}{doi: \begingroup \urlstyle{rm}\Url}\fi

\bibitem[Anderson(1971)]{anderson1971statistical}
T.~W. Anderson.
\newblock \emph{The statistical analysis of time series}.
\newblock John Wiley \& Sons, Inc., New York-London-Sydney, 1971.

\bibitem[Belomestny et~al.(2020)Belomestny, Iosipoi, Moulines, Naumov, and
  Samsonov]{belomestny2020variance}
Denis Belomestny, Leonid Iosipoi, Eric Moulines, Alexey Naumov, and Sergey
  Samsonov.
\newblock Variance reduction for {M}arkov chains with application to {MCMC}.
\newblock \emph{Statistics and Computing}, 30:\penalty0 973--997, 2020.

\bibitem[Berg and Song(2023)]{berg2022efficient}
Stephen Berg and Hyebin Song.
\newblock Efficient shape-constrained inference for the autocovariance sequence
  from a reversible markov chain.
\newblock \emph{The Annals of Statistics}, 51\penalty0 (6):\penalty0
  2440--2470, 2023.

\bibitem[Berg et~al.(2019)Berg, Zhu, and Clayton]{berg2019control}
Stephen Berg, Jun Zhu, and Murray~K. Clayton.
\newblock Control variates and {R}ao-{B}lackwellization for deterministic sweep
  {M}arkov chains.
\newblock \emph{arXiv preprint arXiv:1912.06926}, 2019.

\bibitem[Chakraborty et~al.(2022)Chakraborty, Bhattacharya, and
  Khare]{chakraborty2022estimating}
Saptarshi Chakraborty, Suman~K Bhattacharya, and Kshitij Khare.
\newblock {Estimating accuracy of the MCMC variance estimator: Asymptotic
  normality for batch means estimators}.
\newblock \emph{Statistics \& Probability Letters}, 183:\penalty0 109337, 2022.

\bibitem[Choi and Hobert(2013)]{choi2013polya}
Hee~Min Choi and James~P. Hobert.
\newblock {The Polya-Gamma Gibbs sampler for Bayesian logistic regression is
  uniformly ergodic}.
\newblock \emph{Electronic Journal of Statistics}, 7:\penalty0 2054 -- 2064,
  2013.

\bibitem[Dai and Jones(2017)]{dai2017multivariate}
Ning Dai and Galin~L. Jones.
\newblock {Multivariate initial sequence estimators in Markov chain Monte
  Carlo}.
\newblock \emph{Journal of Multivariate Analysis}, 159:\penalty0 184--199,
  2017.

\bibitem[Damerdji(1991)]{Damerdji1991-oj}
Halim Damerdji.
\newblock Strong consistency and other properties of the spectral variance
  estimator.
\newblock \emph{Management Science}, 37:\penalty0 1424--1440, 1991.

\bibitem[Dellaportas and Kontoyiannis(2012)]{dellaportas2012control}
Petros Dellaportas and Ioannis Kontoyiannis.
\newblock Control variates for estimation based on reversible {M}arkov chain
  {M}onte {C}arlo samplers.
\newblock \emph{Journal of the Royal Statistical Society: Series B},
  74:\penalty0 133--161, 2012.

\bibitem[Flegal and Jones(2010)]{flegal2010batch}
James~M. Flegal and Galin~L. Jones.
\newblock {Batch means and spectral variance estimators in Markov chain Monte
  Carlo}.
\newblock \emph{The Annals of Statistics}, 38:\penalty0 1034 -- 1070, 2010.

\bibitem[Flegal et~al.(2021)Flegal, Hughes, Vats, Dai, Gupta, and
  Maji]{mcmcse_R}
James~M. Flegal, John Hughes, Dootika Vats, Ning Dai, Kushagra Gupta, and
  Uttiya Maji.
\newblock \emph{{mcmcse: Monte Carlo Standard Errors for MCMC}}, 2021.
\newblock R package version 1.5-0.

\bibitem[Folland(1999)]{folland1999real}
Gerald~B Folland.
\newblock \emph{Real analysis: modern techniques and their applications},
  volume~40.
\newblock John Wiley \& Sons, 1999.

\bibitem[Geyer(1992)]{geyer1992practical}
Charles~J. Geyer.
\newblock {Practical Markov chain Monte Carlo}.
\newblock \emph{Statistical Science}, 7:\penalty0 473--483, 1992.

\bibitem[Heidelberger and Welch(1981)]{heidelberger1981spectral}
Philip Heidelberger and Peter~D Welch.
\newblock A spectral method for confidence interval generation and run length
  control in simulations.
\newblock \emph{Communications of the ACM}, 24:\penalty0 233--245, 1981.

\bibitem[Hoffman et~al.(2014)Hoffman, Gelman, et~al.]{hoffman2014no}
Matthew~D Hoffman, Andrew Gelman, et~al.
\newblock The {N}o-{U}-{T}urn sampler: adaptively setting path lengths in
  {H}amiltonian {M}onte {C}arlo.
\newblock \emph{Journal of Machine Learning Research}, 15:\penalty0 1593--1623,
  2014.

\bibitem[Jones(2004)]{jones2004central}
Galin~L. Jones.
\newblock {On the Markov chain central limit theorem}.
\newblock \emph{Probability Surveys}, 1:\penalty0 299 -- 320, 2004.

\bibitem[Jones et~al.(2006)Jones, Haran, Caffo, and Neath]{jones2006fixed}
Galin~L Jones, Murali Haran, Brian~S Caffo, and Ronald Neath.
\newblock {Fixed-width output analysis for Markov chain Monte Carlo}.
\newblock \emph{Journal of the American Statistical Association}, 101:\penalty0
  1537--1547, 2006.

\bibitem[Kontoyiannis and Meyn(2012)]{kontoyiannis2012geometric}
Ioannis Kontoyiannis and Sean~P Meyn.
\newblock {Geometric ergodicity and the spectral gap of non-reversible Markov
  chains}.
\newblock \emph{Probability Theory and Related Fields}, 154:\penalty0 327--339,
  2012.

\bibitem[Kosorok(2000)]{kosorok2000error}
Michael~R. Kosorok.
\newblock {Monte Carlo error estimation for multivariate Markov chains}.
\newblock \emph{Statistics \& Probability Letters}, 46:\penalty0 85--93, 2000.

\bibitem[Liu et~al.(2021)Liu, Vats, and Flegal]{liu2021batch}
Ying Liu, Dootika Vats, and James~M Flegal.
\newblock Batch size selection for variance estimators in {MCMC}.
\newblock \emph{Methodology and Computing in Applied Probability}, pages 1--29,
  2021.

\bibitem[Martin et~al.(2011)Martin, Quinn, and Park]{JSSv042i09}
Andrew~D. Martin, Kevin~M. Quinn, and Jong~Hee Park.
\newblock {MCMCpack}: Markov chain {M}onte {C}arlo in {R}.
\newblock \emph{Journal of Statistical Software}, 42:\penalty0 1–21, 2011.

\bibitem[Meyn and Tweedie(2009)]{meyn2009markov}
Sean Meyn and Richard~L. Tweedie.
\newblock \emph{{M}arkov Chains and Stochastic Stability}.
\newblock Cambridge University Press, Cambridge, second edition, 2009.

\bibitem[{\=O}sawa(1988)]{osawa1988reversibility}
Hideo {\=O}sawa.
\newblock Reversibility of first-order autoregressive processes.
\newblock \emph{Stochastic processes and their applications}, 28:\penalty0
  61--69, 1988.

\bibitem[Polson et~al.(2013)Polson, Scott, and Windle]{polson2013Bayesian}
Nicholas~G. Polson, James~G. Scott, and Jesse Windle.
\newblock Bayesian inference for logistic models using {P}ólya–{G}amma
  latent variables.
\newblock \emph{Journal of the American Statistical Association}, 108:\penalty0
  1339--1349, 2013.

\bibitem[Priestley(1981)]{priestley1981spectral}
M.B. Priestley.
\newblock \emph{Spectral Analysis and Time Series}.
\newblock Number Bd. 1-2 in Probability and mathematical statistics : a series
  of monographs and textbooks. Academic Press, 1981.

\bibitem[{R Core Team}(2020)]{rlang}
{R Core Team}.
\newblock \emph{R: A Language and Environment for Statistical Computing}.
\newblock R Foundation for Statistical Computing, Vienna, Austria, 2020.

\bibitem[Robert and Casella(2004)]{robert2004monte}
Christian~P. Robert and George Casella.
\newblock \emph{Monte {C}arlo Statistical Methods}.
\newblock Springer-Verlag, New York, second edition, 2004.

\bibitem[Roberts and Rosenthal(1997)]{Roberts1997-qk}
Gareth Roberts and Jeffrey Rosenthal.
\newblock {Geometric ergodicity and hybrid Markov chains}.
\newblock \emph{Electronic Communications in Probability}, 2:\penalty0 13--25,
  1997.

\bibitem[Roberts and Rosenthal(1998)]{roberts1998markov}
Gareth~O Roberts and Jeffrey~S Rosenthal.
\newblock Markov-chain {M}onte {C}arlo: some practical implications of
  theoretical results.
\newblock \emph{{The Canadian Journal of Statistics/La Revue Canadienne de
  Statistique}}, 26:\penalty0 5--20, 1998.

\bibitem[South et~al.(2022)South, Riabiz, Teymur, and
  Oates]{south2022postprocessing}
Leah~F South, Marina Riabiz, Onur Teymur, and Chris~J Oates.
\newblock Postprocessing of {MCMC}.
\newblock \emph{Annual Review of Statistics and Its Application}, 9:\penalty0
  529--555, 2022.

\bibitem[{Stan Development Team}(2019)]{stan2019}
{Stan Development Team}.
\newblock Stan reference manual, version 2.28, 2019.

\bibitem[Stein and Shakarchi(2009)]{stein2009real}
Elias~M Stein and Rami Shakarchi.
\newblock \emph{Real analysis: measure theory, integration, and Hilbert
  spaces}.
\newblock Princeton University Press, 2009.

\bibitem[Vats et~al.(2018)Vats, Flegal, and Jones]{vats2018strong}
Dootika Vats, James~M. Flegal, and Galin~L. Jones.
\newblock {Strong consistency of multivariate spectral variance estimators in
  Markov chain Monte Carlo}.
\newblock \emph{Bernoulli}, 24:\penalty0 1860 -- 1909, 2018.

\bibitem[Vats et~al.(2019)Vats, Flegal, and Jones]{vats2019multivariate}
Dootika Vats, James~M Flegal, and Galin~L Jones.
\newblock {Multivariate output analysis for Markov chain Monte Carlo}.
\newblock \emph{Biometrika}, 106:\penalty0 321--337, 2019.

\bibitem[Yu et~al.(2015)Yu, Wang, and Samworth]{yu2015useful}
Yi~Yu, Tengyao Wang, and Richard~J Samworth.
\newblock A useful variant of the {D}avis--{K}ahan theorem for statisticians.
\newblock \emph{Biometrika}, 102:\penalty0 315--323, 2015.

\end{thebibliography}
\fi
\fi
\end{document}